\documentclass{IEEEoj}
\usepackage{cite}
\usepackage{amsmath,amssymb,amsfonts,amsthm}
\usepackage{algorithmic}
\usepackage{graphicx,color}
\usepackage{textcomp}
\usepackage{subfig}

\usepackage{mathtools,bbm}
\allowdisplaybreaks

\usepackage[capitalise]{cleveref}

\usepackage{afterpage}

\newtheorem{thm}{Theorem}
\newtheorem{lem}[thm]{Lemma}

\DeclarePairedDelimiterXPP\E[1]{\mathop{\mathbb{E}}}{[}{]}{}{

#1
}

\DeclarePairedDelimiterXPP\Pro[1]{\mathop{\mathbb{P}}}{[}{]}{}{

#1
}
\DeclarePairedDelimiterXPP\Ex[2]{\mathop{\mathbb{E}}_{#1}}{[}{]}{}{

#2
}
\DeclarePairedDelimiterXPP\Var[2]{\mathop{\mathbb{Var}}_{#1}}{[}{]}{}{

#2
}
\DeclarePairedDelimiterXPP\Prob[2]{\mathop{\mathbb{P}}_{#1}}{[}{]}{}{

#2
}

\def\BibTeX{{\rm B\kern-.05em{\sc i\kern-.025em b}\kern-.08em
    T\kern-.1667em\lower.7ex\hbox{E}\kern-.125emX}}
\AtBeginDocument{\definecolor{ojcolor}{cmyk}{0.93,0.59,0.15,0.02}}
\def\OJlogo{\vspace{-4pt}\hskip-4pt\includegraphics[height=18pt]{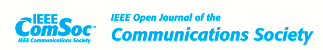}}
\begin{document}
\receiveddate{XX Month, XXXX}
\reviseddate{XX Month, XXXX}
\accepteddate{XX Month, XXXX}
\publisheddate{XX Month, XXXX}
\currentdate{12 June, 2025}
\doiinfo{OJCOMS.2024.011100}

\title{Wireless Network Topology Inference: A Markov Chains Approach}

\author{JAMES MARTIN\IEEEauthorrefmark{2} , TRISTAN PRYER \IEEEauthorrefmark{1,2}, AND LUCA ZANETTI\IEEEauthorrefmark{1,2}}
\affil{Department of Mathematical Sciences, University of Bath, Bath, BA2 7AY UK}
\affil{Institute for Mathematical Innovation, University of Bath, Bath, BA2 7AY UK}
\corresp{CORRESPONDING AUTHOR: Luca Zanetti (e-mail: lz2040@bath.ac.uk).}
\authornote{All authors were partially supported by the Defence Science and Technology Laboratory.}
\markboth{Preparation of Papers for IEEE OPEN JOURNALS}{Author \textit{et al.}}

\begin{abstract}
We address the problem of inferring the topology of a
wireless network using limited observational data. Specifically, we
assume that we can detect when a node is transmitting, but no further
information regarding the transmission is available. We propose a
novel network estimation procedure grounded in the following abstract
problem: estimating the parameters of a finite discrete-time Markov
chain by observing, at each time step, which states are visited by
multiple ``anonymous'' copies of the chain. We develop a consistent
estimator that approximates the transition matrix of the chain in the
operator norm, with the number of required samples scaling roughly
linearly with the size of the state space. Applying this estimation
procedure to wireless networks, our numerical experiments demonstrate
that the proposed method accurately infers network topology across a
wide range of parameters, consistently outperforming transfer entropy,
particularly under conditions of high network congestion.
\end{abstract}


\maketitle

\section{INTRODUCTION}
Inferring the topology of a network from limited information has
numerous applications across various fields, including the discovery
of gene interactions in computational biology, the analysis of climate
dynamics in environmental science, and the investigation of functional
brain structure in neuroscience~\cite{net_inf_survey, Antonacci20, Antonacci24}. Additionally,
in defense and security contexts, understanding the structure of an
adversary's network can provide valuable insights for decision-making.

In many practical scenarios, the available information about a network
often consists of time series data associated with each
node. Specifically, in wireless networks, radio frequency sensors
deployed in the environment may allow us to detect when each node is
transmitting, but not the recipient of the transmission. The objective
is then to infer whether pairs of nodes are connected by exploiting
correlations in their corresponding time series.

Previous work, particularly in the context of wireless
networks~\cite{TR13,LC17,TG20,liu2022topology}, has focused on
detecting links by computing causality measures between time series,
such as Granger causality~\cite{granger69} and transfer
entropy~\cite{schreiber00}. While these approaches have demonstrated
reasonable accuracy, particularly in simulations, they remain poorly
understood. For example, it is unclear in what ways the inferred
network approximates the actual network structure.

In this work, we introduce a novel paradigm for network
inference. Rather than focusing on local link detection and
subsequently attempting to infer global network information by
assembling local estimates, we aim to construct a global approximation
that better captures the overall structure of the network.

Our goal is to develop a procedure that \emph{provably} approximates
the observed network under a well-defined, albeit simplified, model of
network dynamics.

\subsection{Our Contribution}
We address the problem of estimating the topology of a wireless
network from limited information. Specifically, we monitor a network
of wireless devices, assuming we can detect when a device is
transmitting, but not the target of the transmission. Such data could
be collected via a network of radio frequency sensors (see \cref{fig:network}
for a schematic depiction).

\begin{figure}[h!]
\centering
   \includegraphics[width=0.25\textwidth]{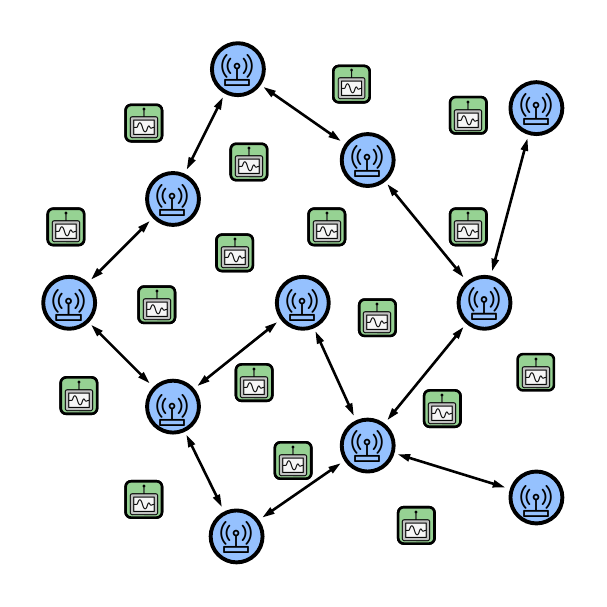}
    \caption{
    A network of wireless devices (blue circles) monitored by a collection of radio frequency sensors
     (green squares). Whenever a wireless device is transmitting, nearby sensors will measure an associated power bursts.
      Our goal is to infer the topology of the network using only the sensors' measurements. 
    }
    \label{fig:network}
\end{figure}

To develop our algorithm, we adopt a Markovian assumption regarding
the network's dynamics: we assume that the future state of the network
depends only on its current state and is independent of its past
states. While this assumption is not strictly valid for communication
networks, we argue that it simplifies the model sufficiently to allow
rigorous analysis while still capturing the essential behaviour of
wireless networks.

More precisely, we consider the following abstract model: we observe
the trajectories of $k \geq 1$ independent discrete-time Markov chains
on a finite state space $V$, which represents the set of nodes in the
network. Each chain evolves according to a common transition matrix
$P$. We assume that we can detect whether a chain is visiting a node
at time $t$, but not which specific chain is present. This models the
scenario in wireless networks where only node transmissions are
observable at time $t$. In this context, $P(u,v)$ denotes the
probability that node $v$ will transmit at time $t+1$, given that node
$u$ was transmitting at time $t$.

Our primary theoretical contribution is the development of a
consistent estimator. We prove that, by observing the trajectories of
the chains for a sufficient duration, we can construct an estimated
transition matrix $\hat{P}$ that approximates the true matrix $P$
arbitrarily well in the operator norm. Crucially, this approximation
in operator norm provides insights into the global structure of the
network, such as the presence of clusters and
bottlenecks~\cite{STsparse}.

We apply this estimation procedure to wireless networks simulated
using the $\texttt{ns-3}$ network simulator \cite{ns3}. Our
experimental results show that the proposed method is competitive with
state-of-the-art techniques. Notably, we demonstrate that the
estimated topology closely approximates the true topology for various
networks with differing structures and sizes. Even in cases where the
number of observed transmissions is too small for most links to be
detected, our procedure successfully recovers structural information
about the network, such as the presence of bottlenecks.

\subsection{Related Work}
From a theoretical standpoint, our work relates to the broader problem
of estimating Markov chain parameters from observed trajectories, an
area that has garnered significant attention in recent
years~\cite{hsu19,Wolf19,WolfKont24}. Our algorithm builds upon
techniques developed by Wolfer and Kontorovich~\cite{Wolf19}.
Prior work, however, typically focuses on scenarios where only a single Markov
chain is observed, which simplifies the estimation task but is less
applicable to modelling topology inference in wireless networks. In
the single-chain case, observing consecutive transmissions at nodes
$u$ and $v$ directly implies the existence of a link between them, a
property that does not hold in our model.

There is a significant body of literature on wireless network topology
inference (see, e.g.,
~\cite{TR13,LC17,salmond19,TG20,mehrotra2021minimax,liu2022topology,du2023network}). However,
to the best of our knowledge, our work is the first to provide
theoretical guarantees for the estimation procedure, at least under a
simplified model of network dynamics. While this model may not fully
capture the complexities of real wireless networks, we believe it
encapsulates the key challenges inherent to the inference task.

Finally, given the novelty of our approach, this work represents an
initial step in a new line of research on wireless network
inference. Further investigation is required to fully assess the
strengths and limitations of the methodology. We discuss potential
directions for future research in the conclusion section.

\subsection{Organisation}
We introduce the setup and the formal model underpinning our inference procedure in \cref{sec:setup}, which then is described in \cref{sec:estimation}. 
In \cref{sec:theory} we state and discuss our main theoretical result, while its proof is deferred to \cref{sec:proof}. We showcase and discuss our experimental results in \cref{sec:exp}. In particular, we first consider an idealised scenario in which we know exactly when a given node is transmitting (\cref{sec:exp} \ref{sec:ideal}) and we later discuss how to acquire this information using radio frequency sensors (\cref{sec:exp} \ref{sec:real}). We end with a discussion of possible future work in \cref{sec:conclusion}. Code and datasets to reproduce our experiments are available at~\cite{ourcode}.

\section{PROBLEM SETUP}
\label{sec:setup}

Let $G=(V,E)$ be a graph representing the wireless network we would like to infer: $V = \{1, \dots,n\}$ represents the set of nodes of the network, while $E \subseteq V \times V$ represents the set of edges (links). We assume $V$ is known and our goal is to infer $E$.

We discretise the time into $T$ intervals. For any time interval $t \in \{1,\dots,T\}$ and any node $u \in V$, we observe if $u$ has emitted a signal during the time interval $t$. We collect this information in a \emph{time series matrix} $TS \in \{0,1\}^{n \times T}$ such that, for any $u \in V$ and $t \in \{1,\dots,T\}$,
\[
TS(u,t) = \begin{cases}
1 & \text{if $u$ is transmitting during time interval $t$} \\
0 & \text{otherwise}.
\end{cases}
\]
Our main task is to infer the set of edges $E$ given the time series matrix $TS$.

We remark that, in practical scenarios, we do not usually have access to perfect information about when a device is transmitting. Rather, we can only obtain an approximation of the time series matrix, for example, by deploying a set of radio frequency sensors in the environment. We expand on how RF sensors can be used to approximate this time series matrix in \cref{sec:exp} \ref{sec:real}.


To formally model network dynamics, we consider $k$ independent Markov chains $X^{(1)},X^{(2)},\dots,X^{(k)}$ on a finite state space $V$ (the set of nodes of our network) and with the same transition matrix $P$. We denote with $X^{(i)}_t$ the position of the $i$-th chain at time $t$. In particular, $X^{(i)}_t = u$ means that, at time $t$, node $u$ is transmitting a message to another node. $P(u,v)$ represents the probability that $v$ transmits a message at time $t+1$, given that a message has been transmitted by $u$ at time $t$. 
Our aim is to infer the transition matrix $P$: this will allows us to infer information not only about the topology of the network, but also information about the role of certain nodes in the network. 

For technical reasons, we assume $P$ has a unique stationary distribution $\pi$, i.e., $\pi P = \pi$. This is a reasonable assumption in our context since it essentially means the network is (strongly) connected, i.e., any two nodes of the network can communicate with each other (without necessarily being connected by a link). Furthermore, we assume the Markov chains are stationary, i.e., $X^{(i)}_1 \sim \pi$ for all $1\le i \le k$. 
This technical assumption is made to simplify our analysis, but is not necessary and could be dropped with additional technical work (without modifying the proposed algorithm).

We assume $P$ and $\pi$ are unknown to us. Indeed, our goal is to estimate $P$. What we are able to observe is, for each time step $t=1,\dots,T$, how many chains are currently visiting a given vertex $v \in V$. In other words, for any $t=1,\dots,T$ and $v \in V$, we can observe the quantity $S_{t,v} = |\{i \colon X^{(i)}_t = v\}|$. We are not able, however, to observe if a specific chain is visiting $v$ at time $t$. 

\section{ESTIMATION PROCEDURE}
\label{sec:estimation}

We are now ready to state our inference procedure. Let $u,v \in V$. We use $[n]$ to denote $\{1,\dots,n\}$. We define the following quantities.
\begin{align*}
N(u,v) &= \left|\{i,j\in [k], t \in [T] \colon X^{(i)}_{t-1} = u, X^{(j)}_t = v \}\right|,\\
N(u) &= \left|\{ i \in [k], 1 \le t < T \colon X^{(i)}_{t} = u\}\right|.
\end{align*}

In other words, $N(u,v)$ represents the number of times we have observed $u$'s transmission being immediately followed by $v$'s, while $N(u)$ counts the total number of times $u$ is transmitting.
Notice that $N(u) = \sum_{v \in V} N(u,v)$. Moreover, these quantities can be computed having access only to $S_{t,v}$ for $1\le t \le T, v \in V$.

Let also $M \in \mathbb{R}^{n \times n}$ be defined as
\begin{equation*}
M(u,v) = \frac{N(u,v)}{N(u)}.
\end{equation*}
Essentially, $M(u,v)$ represents the proportion of $u$'s transmissions that are immediately followed by $v$'s.

Let $\widehat{\Pi}$ be the $n \times n$ matrix whose entries are equal to $\widehat{\Pi}(u,v) = N(v)/(kT)$ for any $u,v \in V$. Our matrix estimator $\hat{P}$ is defined as follows:
\[
\hat{P} = M - (k-1)\widehat{\Pi}.
\]

The term $(k-1)\widehat{\Pi}$ is a correction necessary because nodes that are visited by the chains more frequently will have larger entries in $M$; in particular, if the chains often visit two nodes $u$ and $v$, the corresponding entry $M(u,v)$ will be large even though there might not be an edge between $u$ and $v$. Indeed, we expect that, on average, $M(u,v) \approx (k-1)\widehat{\Pi}$ whenever $u$ is \emph{not} connected to $v$.

Computing $\hat{P}$ is computationally efficient and can be done in an online manner: we do not have to store the entire time series matrix $TS$, we can just update $N(u,v)$, $N(v)$, and $\hat{P}$  at each time step. This can be done in time $O(n^2 T)$ in total.

Returning to the wireless network setting, $\hat{P}(u,v)$ can be interpreted as our estimated probability that $v$ will transmit immediately after $u$ has transmitted. While $\hat{P}(u,v)$ is strongly correlated to the likelihood $u$ and $v$ are connected by a link, it is not a measure of how likely a link between $u$ and $v$ exists. In particular, if a node $u$ has many outgoing links, entries in the corresponding row $\hat{P}(u,\cdot)$ will necessarily be small.

To overcome this issue, we suggest using an estimator for the Laplacian associated to $P$. For a transition matrix $P$ with stationary distribution $\pi$, the Laplacian $L$ of $P$ is the matrix with entries $L(u,v) =  P(u,v) \sqrt{\pi(v)/ \pi(u)}$. We recall that any Markov chain can be interpreted as a random walk on a weighted directed graph, where $L(u,v)$ represents the weight on the edge $(u,v)$~\cite{markovmixing}. Therefore, to infer if a link is present or not, we will use the estimated matrix $\hat{L}$, defined as
\[
\hat{L}(u,v) =  \sqrt{\frac{\hat{\pi}(v)}{\hat{\pi}(u)}} \hat{P}(u,v)
\]
where $\hat{\pi}$ is the leading eigenvector of $\hat{P}$ and an estimate for the stationary distribution $\pi$. 

When we know that  links in the observed network are bidirectional, it is natural to consider a \emph{symmetrisation} of the Laplacian~\cite{MonteTetali} instead:

\[
\hat{L}_{\text{sym}}(u,v) =   \sqrt{\frac{\hat{\pi}(v)}{\hat{\pi}(u)}} \cdot  \frac{\hat{P}(u,v) + \hat{P}(v,u)}{2}.
\]  

We conclude this section by underlining two issues that arise when applying our methodology to infer the topology of real-world wireless networks. First of all, notice that our estimator only records transmissions happening in two consecutive time intervals. This means it is important to choose an appropriate interval length, otherwise network delays might severely impact our estimator. In our simulations, however, this constraint has not been detrimental to our results. The situation might be more complicated in real world settings. 
Secondly, our estimator depends on the number $k$ of Markov chains that are concurrently active. In practice, $k$ can be interpreted as the average number of devices that transmit in the same time interval. We will estimate this parameter as the ratio between the number of nonzero entries in the time series matrix $TS$ and the number of time-intervals $t$ such that there has been a transmission both at interval $t$ and at the following interval $t+1$ (this is because our estimator $\hat{P}$ only considers consecutive transmissions).

\section{THEORETICAL GUARANTEES}
\label{sec:theory}

To present the main theoretical result of our work, we introduce some additional technical notation. Recall that $n=|V|$ is the number of nodes in our network. We say an event $\mathcal{E}$ holds with high probability if the probability of $\mathcal{E}$ happening is larger than $1-n^{-1}$.  We denote with $\sigma_2(P)$ the second largest singular value of $P$ and with $\lambda(P) = 1-\sigma_2(P)^2$ the \emph{spectral gap} of $P$. Given a matrix $M$, we denote with $\|M\|$ its operator norm. We are now ready to state our main theorem.

\begin{thm}
\label{thm:main}
Let $\epsilon \in (0,1/2)$. Define $\pi_\star = \min_{u\in V} \pi(u)$ and $\pi^\star = \max_{u \in V} \pi(u)$. Assume that 
\[
T > C k^3\epsilon^{-2} \ln(kn) \cdot \lambda(P)^{-1}\frac{\pi^\star}{\pi_\star^2}
\]
for some large enough universal constant $C > 0$. Then, with high probability, 
\[
\|P - (M - (k-1)\widehat{\Pi})\| \le \epsilon.
\]
\end{thm}

The theorem essentially tells us that $\hat{P} = M - (k-1)\widehat{\Pi}$ is a good estimator for $P$. In particular, by observing the network for a large enough time, we can approximate $P$ arbitrarily well in the operator norm. This approximation preserves important quantities of the network. For example, it approximately preserves its stationary distribution; this tells us how often, on average, nodes are visited~\cite{markovmixing}, which corresponds how often a node is transmitting in our model. It also preserves its spectral gap, which tells us how fast a Markov chain with transition matrix $P$ converges to its stationary distribution from a worst case starting point~\cite{MonteTetali}. Moreover, it preserves important connectivity properties of the networks, such as its ``cut-structure'', that is, the structure of clusters and bottlenecks in the network~\cite{PSZ,STsparse}.

What this approximation doesn't necessarily preserve are the individual links in the network. For example, consider a transition matrix of a random walk on a complete graph of $n$ vertices. Now consider an Erd\H{o}s-R\'enyi random graph $G(n,p)$: that is, a graph of $n$ vertices where we connect any two vertices with probability $p$. As long as $p \gg log(n)/n$, it is well-known that the transition matrix of a random walk on $G(n,p)$ well approximates (in the operator norm) the transition matrix of the complete graph~\cite{STsparse}. However, the $G(n,p)$ graph might have only $O(n \log(n))$ edges compared to the $\Theta(n^2)$ edges of the complete graph, i.e., only a small fraction of the edges are preserved. Indeed, it is possible to construct examples where Theorem \ref{thm:main} guarantees a good approximation with a number of samples that is much smaller than the number of edges in the graph: certain links wouldn't have been used at all and, therefore, would not be detectable.

We now discuss the sample complexity of our procedure. The number of samples (or time steps) required depends on several parameters. In particular, it is inversely linear  in the spectral gap $\lambda(P)$. The spectral gap is a crucial parameter in the analysis of Markov chains: its reciprocal, $1/\lambda(P)$  provides a bound on the mixing time of $P$, i.e.,  the time it takes for a Markov chain with transition matrix $P$ to converge to the stationary distribution from a worst case starting point~\cite{MonteTetali}. As long as $P$ is irreducible, as in our case, $\lambda(P) > 0$. We believe that, in cases where $P$ models a communication network, $1/\lambda(P)$ is relatively small. Furthermore, for natural transition matrices $P$, $1/\pi_\star$ is proportional to $n$, the size of the network. A larger value of $1/\pi_\star$ would mean, in our model, that there are nodes that transmit very rarely: to approximate $P$, we necessarily need to collect enough information about these nodes, increasing the time required. Finally, $\frac{\pi^\star}{\pi_\star}$ can be interpreted as the ratio between the largest and the smallest frequencies of transmission of a node: if nodes transmit at roughly the same frequency, this parameter becomes order of 1.

%

In summary, for many network topologies, as long as the number of concurrent transmissions $k$ is not too large, Theorem \ref{thm:main} asserts that a number of samples proportional to the size of the network suffices to well approximate $P$.

The proof of Theorem \ref{thm:main}, deferred to Section \ref{sec:proof}, is inspired by the techniques of \cite{Wolf19}. In particular, it uses matrix concentration inequalities for Markov chains \cite{matrixMarkovBernstein} and martingales \cite{Tropp}. 

We end this section by briefly discussing the validity of our Markovian model for capturing our task of estimating the topology of a wireless network. Even though communication networks do not satisfy a Markovian assumption, we hope that, on average, the behaviour of such networks can be approximated well enough by a Markovian model. Besides this consideration, however, for our model to capture the behaviour of wireless networks, an additional assumption must be considered. Since multiple chains can concurrently visit a node at any time, our estimation procedure requires to observe the number of chains visiting a node $u$ at a time $t$. In our modelling, however, a chain visiting a node $u$ at time $t$ corresponds to observing node $u$ transmitting a message during the $t$-th time interval. Therefore, $r$ chains visiting $u$ at time $t$ would correspond to $u$ transmitting $r$ messages during the same time interval. This event, besides being very unlikely, cannot even be detected: we assume we can only observe if a node is transmitting in a certain time interval, not how many messages is transmitting. For this reason, our Markovian model can only be realistic  if the chance of two or more chains visiting the same node at the same time is very small. Formally, the probability that at some time $t$ two stationary chains visit the same node can be upper bounded by
\[
\sum_{u \in V} \sum_{1 \le i < j \le k} \pi(u) = {k \choose 2} \sum_{u \in V} \pi(u)^2 \le k^2 \|\pi\|_2^2.
\]

Therefore, we assume that $k^2 \|\pi\|_2^2 \ll 1$ or, equivalently, that
\begin{equation}
\label{eq:ass_k}
k \ll \frac{1}{\|\pi\|}.
\end{equation}

If the stationary distribution of the chain is approximately uniform, we have that $\|\pi\| \approx 1/\sqrt{n}$. In this case, \cref{eq:ass_k} can be interpreted as $k \ll \sqrt{n}$.

\section{EXPERIMENTAL RESULTS}
\label{sec:exp}

We have simulated various wireless network topologies with the network simulator $\texttt{ns-3}$~\cite{ns3}. Each wireless network is constructed by generating $n$ points on the plane. Each point corresponds to a wireless device and two devices are able to communicate with one another if they are at most 50m apart. Different point sets will induce a different network topology. 
We use the IEEE 802.11b protocol with a data rate of 1Mbps.
We use the OLSR routing protocol to discover and disseminate edge state information. The Hello and TC intervals are 2 and 5 seconds, respectively. We do not send any packets within the first 30 seconds of the simulation to allow time for devices to calculate the initial shortest path between other nodes in the network. We use UDP as the transport layer protocol and each node is implemented as a UDP client. This choice has been made because UDP does not use acknowledgment messages (ACKs), possibly making link detection more difficult. However, we have also performed some simulations with TCP obtaining very similar results.

We sample 500 pairs of sender/receiver nodes $(i,j)$ uniformly at random (with replacement) and transmit 3 packets of size 100 bytes from $i$ to $j$ at time $30 + t$, where $t$ is uniformly distributed in $[0,T]$ for some transmission window length $T>0$.

The time series matrix is constructed by sampling the interval within the simulation between 30 and $30+T$ seconds into intervals of 1.5ms. We remark we haven't really optimised this number: time-intervals between 1ms and 1.5ms usually work quite well, both for the method outlined in this report and for other methods based on transfer entropy or Granger causality.

Our experiments assume the network is bidirectional, and therefore, we will use our estimate  $\hat{L}_{\text{sym}}$ for the symmetrised Laplacian. We will compare our estimate with one derived from transfer entropy. More precisely, we will consider the matrix $M_{\text{TE}}$ where the $(i,j)$-entry is equal to the transfer entropy from the time series associated with $i$ to the time series associated with $j$. To make the comparison even, we will also symmetrise this matrix by taking the average between this matrix and its transpose. Transfer entropy is estimated using the Python library PyInform \cite{PyInform} with embedding dimension $d=5$. We did not optimise such number but found a similar performance for smaller values of $d$, while larger values of $d$ might produce worse results.

\subsection{An idealised setting}
\label{sec:ideal}
We first consider an idealised setting that follows the problem setup described in \cref{sec:setup}: we assume we have direct access to the time series matrix $TS$. While we do not expect to have perfect information about the time series matrix $TS$ in practice, we believe useful to know the strengths and limitations of our approach in this idealised scenario. 

We fix a constant propagation speed (the speed of light) and path loss is dependent only on the transmission range: if the receiver is within range, it will receive the transmission at 16dbm, else it will not receive a transmission. 

In Fig.~\ref{fig:cycle} we present experimental results for a network topology corresponding to an undirected cycle  of $6$ nodes. We can see that both transfer entropy and our methodology works well in recovering links in the network. In Fig.~\ref{fig:cycle}(d) we study the effect of a varying transmission length on network estimation. In particular, we record the proportion of correctly identified links by our methodology and transfer entropy, for a varying length of the transmission window. More precisely, we say that a link $\{i,j\}$ has been correctly identified by our methodology (resp. transfer entropy) if it appears in the top $m$ entries of the upper triangular part of the $\hat{L}_{\text{sym}}$ (resp. transfer entropy) matrix, where $m$ is the number of links in the network.\footnote{We have adopted this simplistic inference procedure only to compare our estimator to transfer entropy. More complicated inference procedures, such as ones based on surrogate data approaches~\cite{pinto2024} or machine learning methods~\cite{TG20}  might be needed in real world scenarios.}
We observe that a shorter transmission window makes estimation harder since it implies a larger number of nodes will be concurrently transmitting (i.e., a larger parameter $k$ in \cref{thm:main}). We notice how our procedure slightly improves upon transfer entropy when the network is congested.

Further experimental results are reported in \cref{fig:cyclestar,fig:3by3grid,fig:ladder,fig:path20}. We highlight that our methodology consistently outperforms transfer entropy whenever the network is congested. Transfer entropy and our estimate appear to work similarly well when $T \ge 15s$, with our estimator, however, giving more weight to pairs of nodes not connected by links: we believe this might be due to an underestimation of the value of $k$ defined in \cref{sec:estimation}. A more refined approach for estimating $k$ might further improve our estimation procedure.

\begin{figure*}[h!]
    \centering
    \subfloat[][Adjacency matrix]{\includegraphics[width=0.25\textwidth]{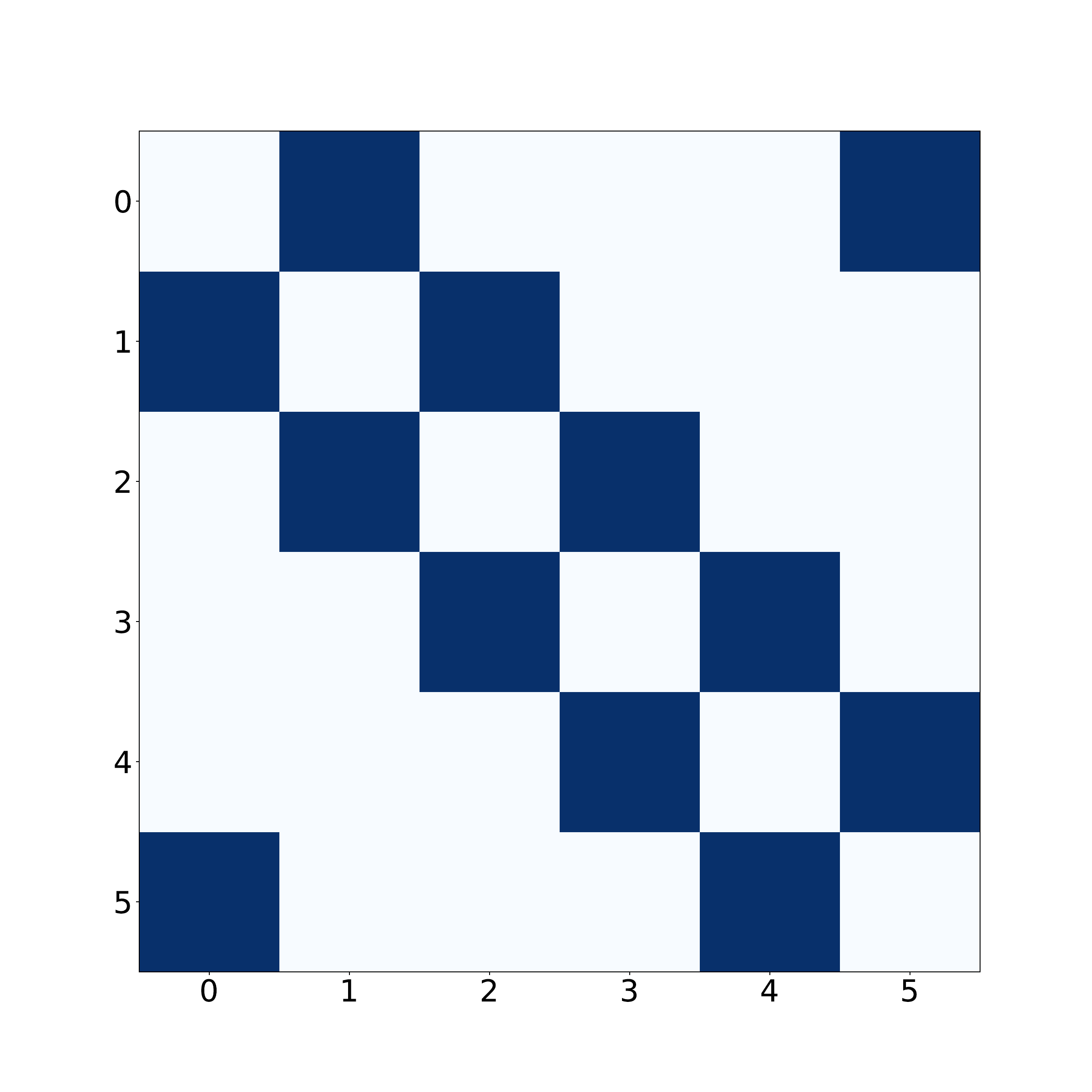}}
    \subfloat[][Our estimate $\hat{L}_{\text{sym}}$]{\includegraphics[width=0.25\textwidth]{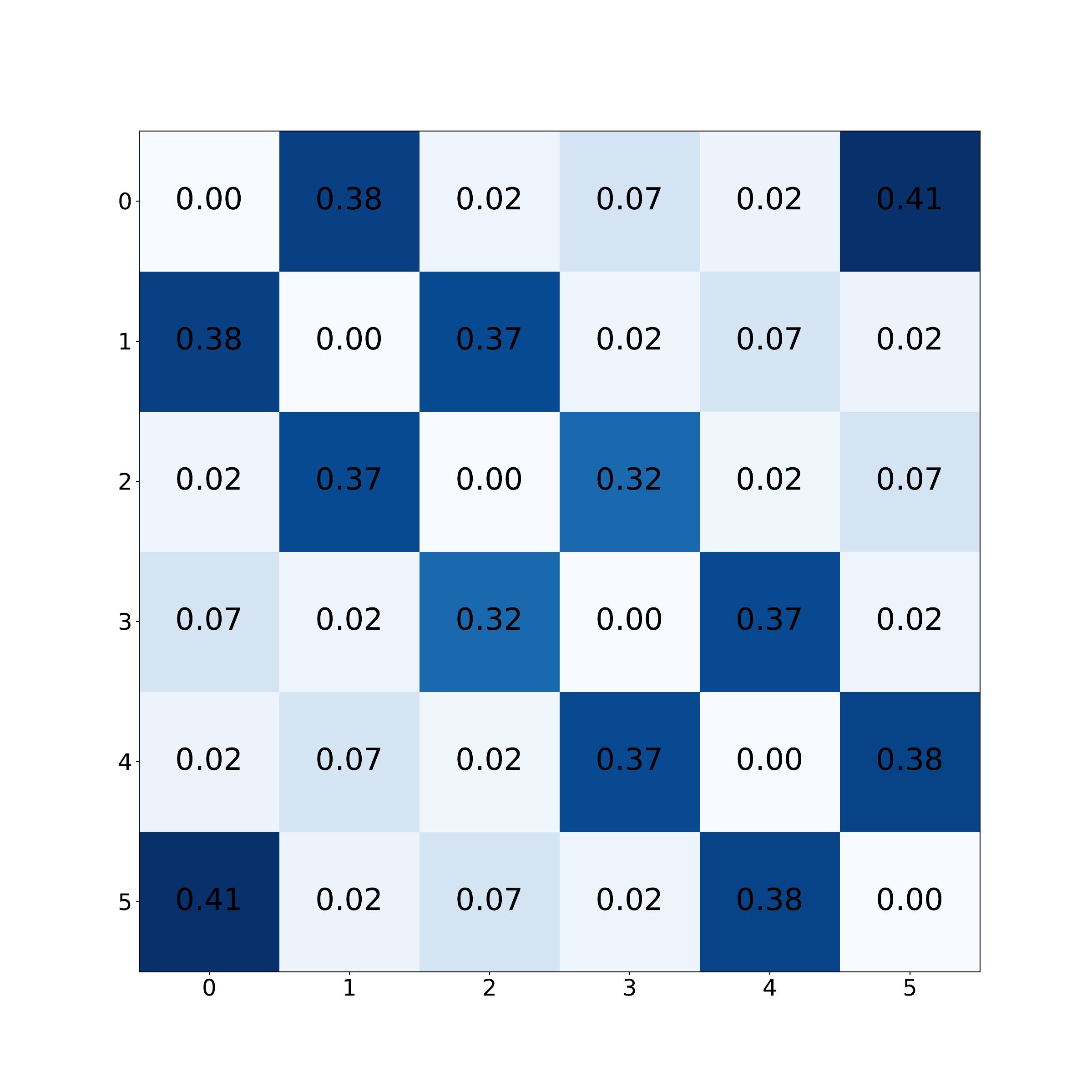}}
    \subfloat[][Transfer entropy matrix]{\includegraphics[width=0.25\textwidth]{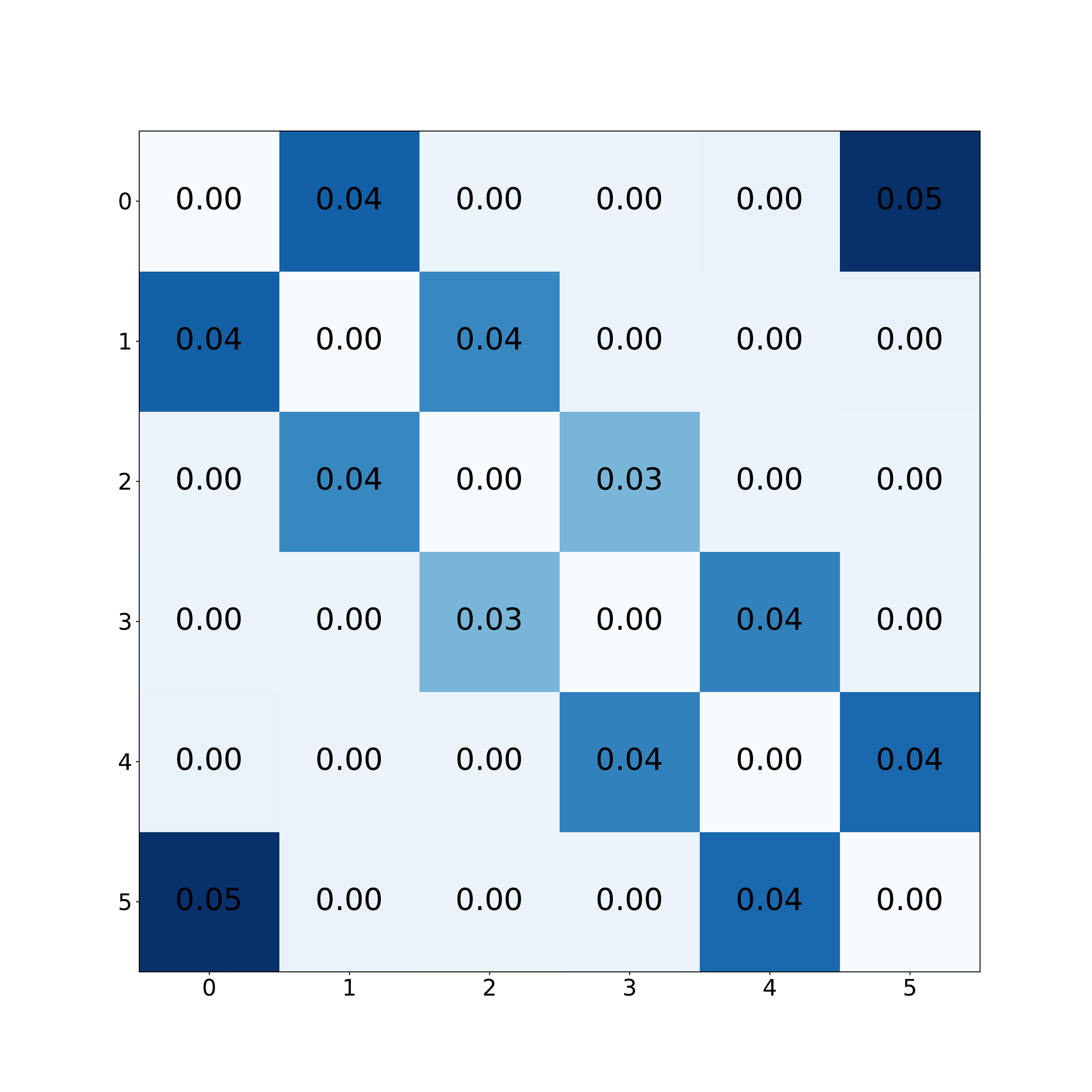}}
    \subfloat[][Varying transmission length]{\includegraphics[width=0.25\textwidth]{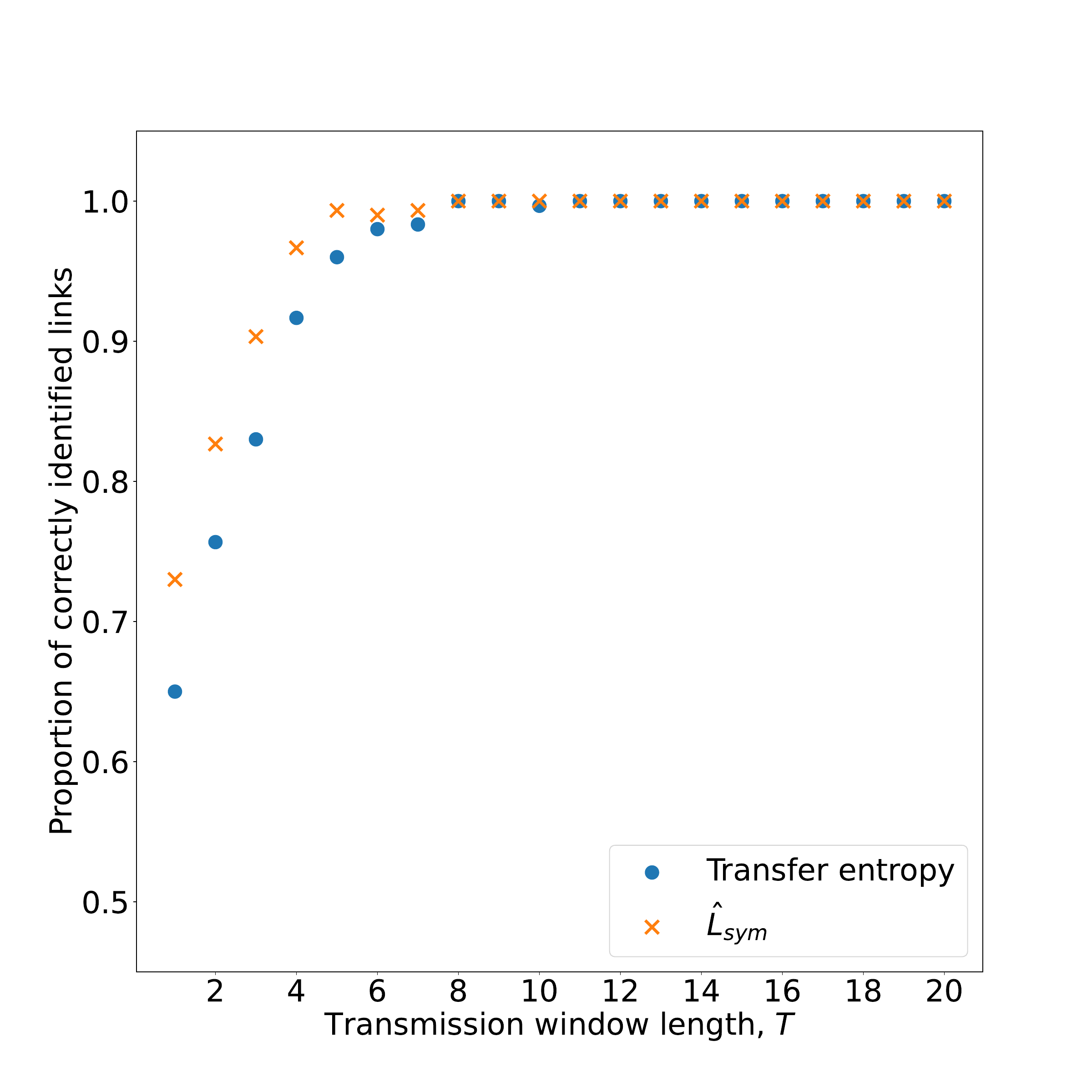}}
    \vspace{0.5cm}
    \caption{
    Experimental results for a cyclic network topology of six nodes. Figure (a) displays the adjacency matrix of the network. Figure (b) the estimated matrix $\hat{L}_{\text{sym}}$ for a transmission window length $T=15s$. Figure (c) the  transfer entropy matrix ($T=15s$). Figure (d) represents the proportion of correctly identified links by our methodology and transfer entropy, for a varying length of the transmission window.  All results are averaged over $50$ independent simulations. A description of the experimental setup is provided at the beginning of \cref{sec:exp}.
    }
    \label{fig:cycle}
\end{figure*}

\begin{figure*}[h!]
    \centering
    \subfloat[][Adjacency matrix]{\includegraphics[width=0.25\textwidth]{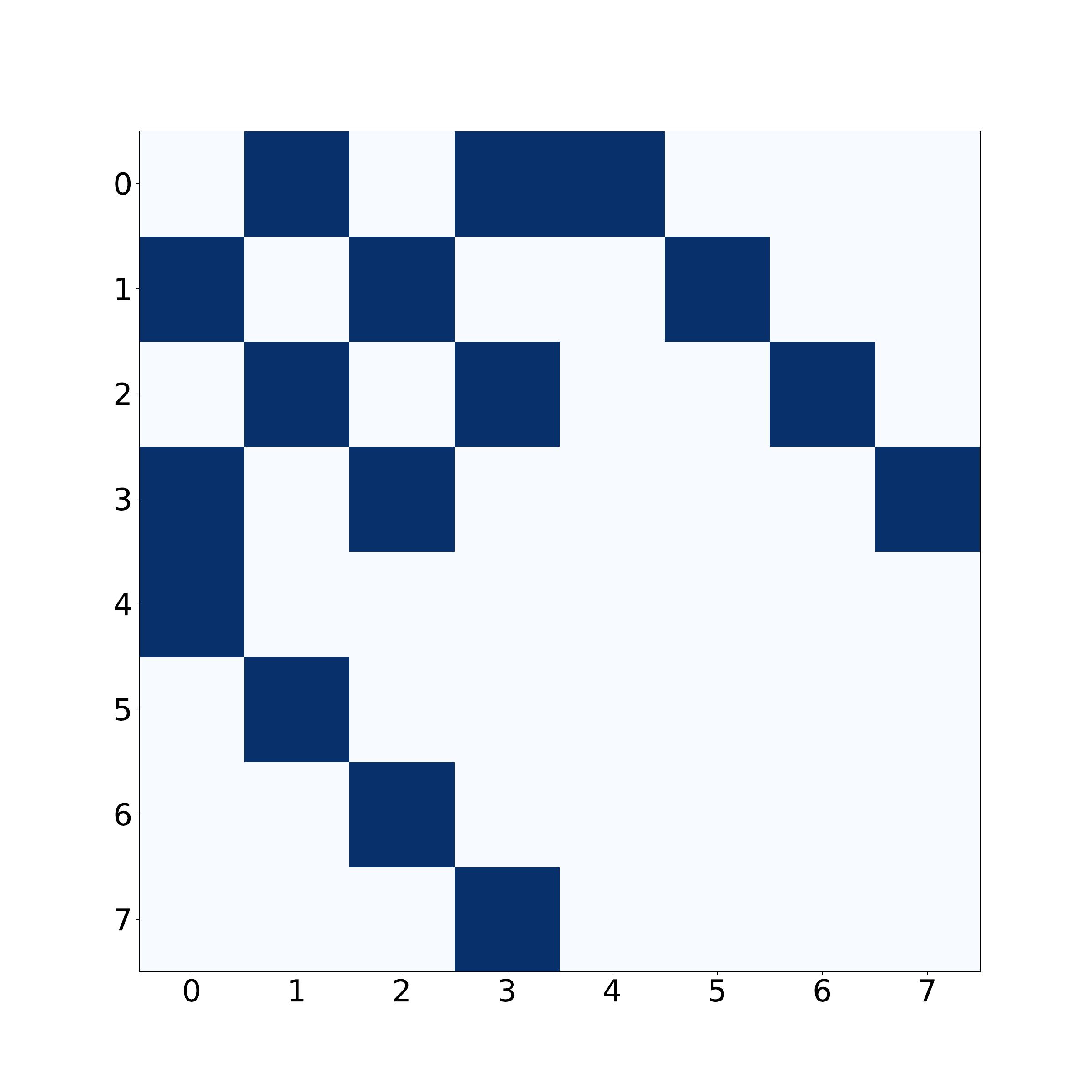}}
    \subfloat[][Our estimate $\hat{L}_{\text{sym}}$]{\includegraphics[width=0.25\textwidth]{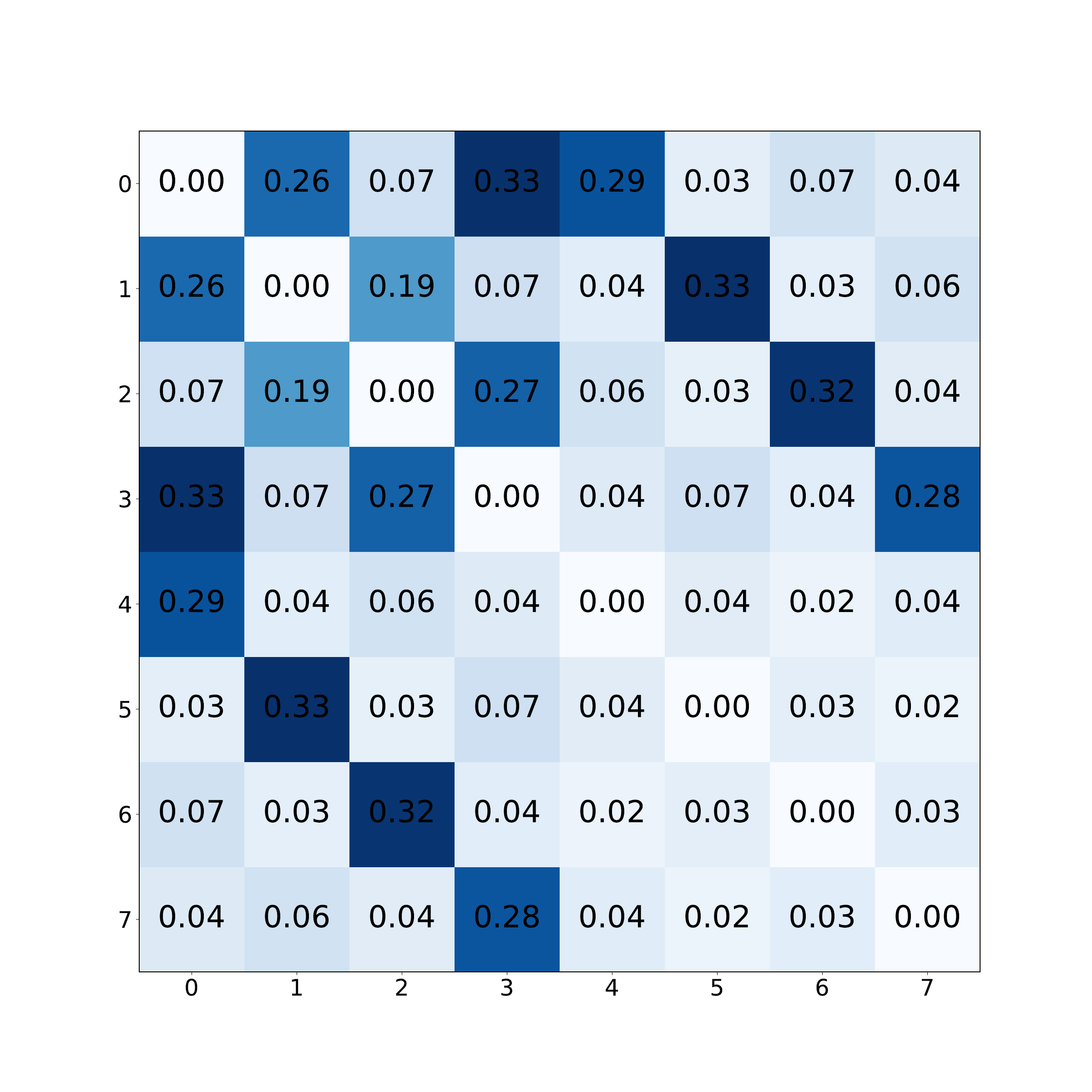}}
    \subfloat[][Transfer entropy matrix]{\includegraphics[width=0.25\textwidth]{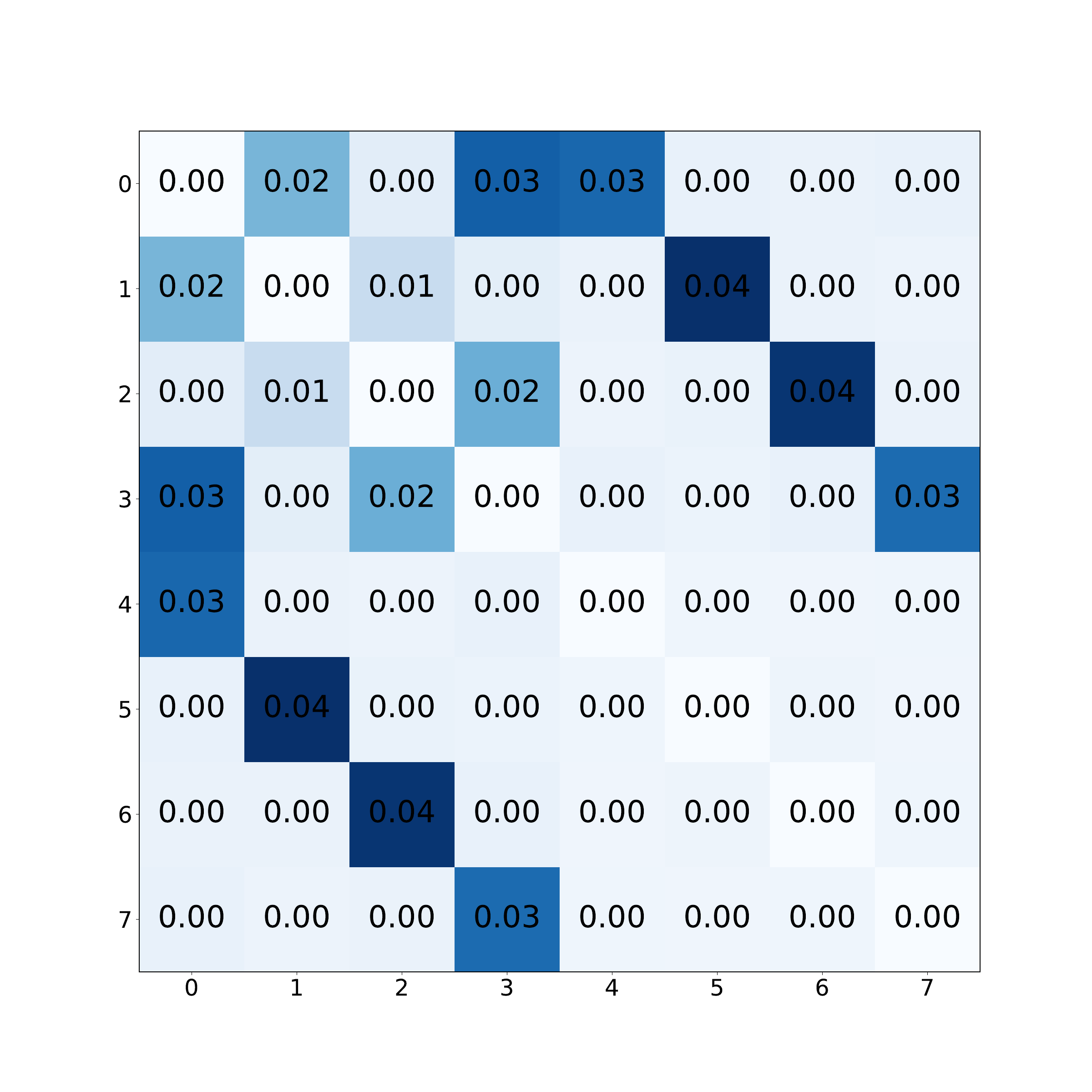}}
    \subfloat[][Varying transmission length]{\includegraphics[width=0.25\textwidth]{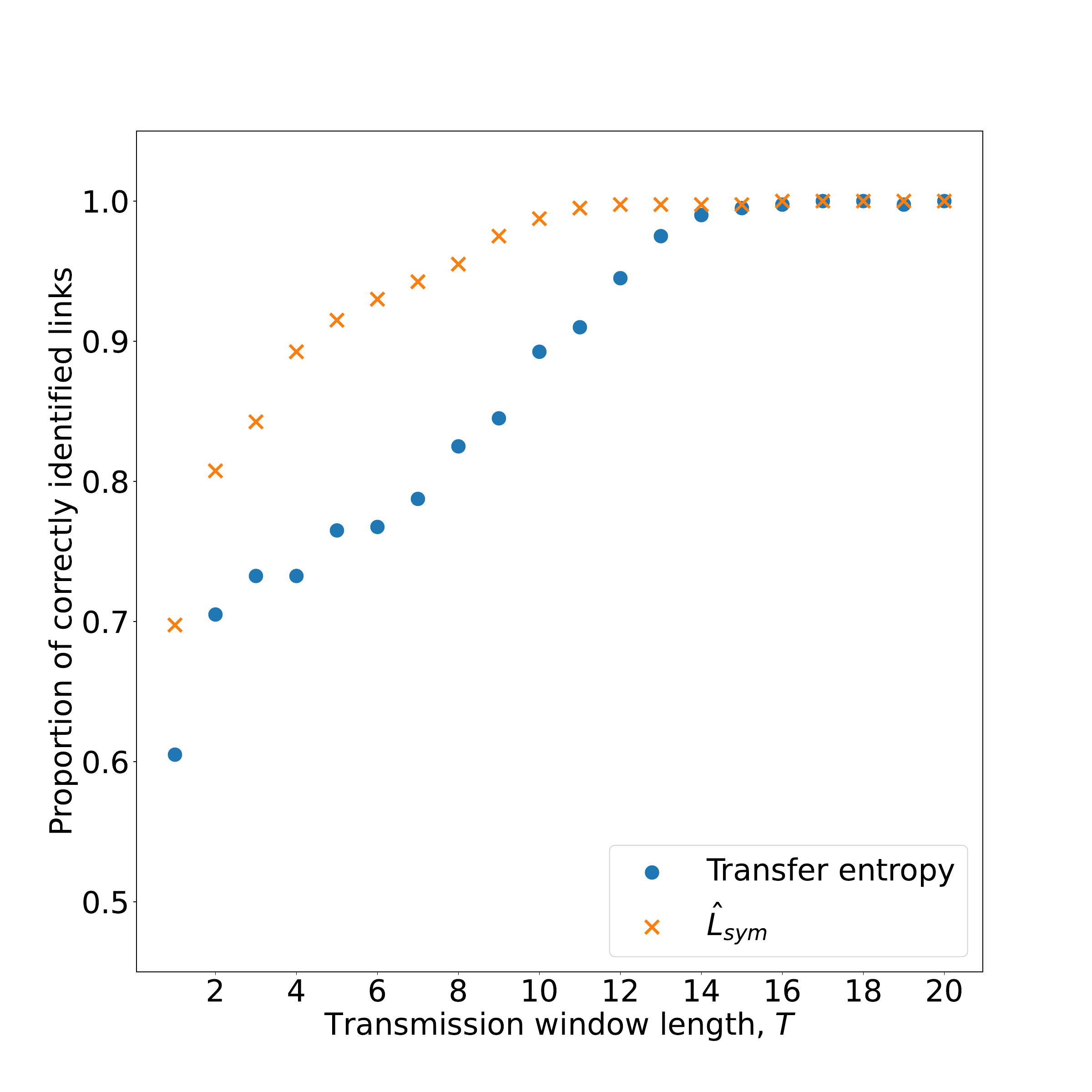}}
    \vspace{0.5cm}
    \caption{
    Experimental results for a network topology of eight nodes. See \cref{fig:cycle} for a description.    }
    \label{fig:cyclestar}
\end{figure*}

\begin{figure*}[h!]
    \centering
    \subfloat[][Adjacency matrix]{\includegraphics[width=0.25\textwidth]{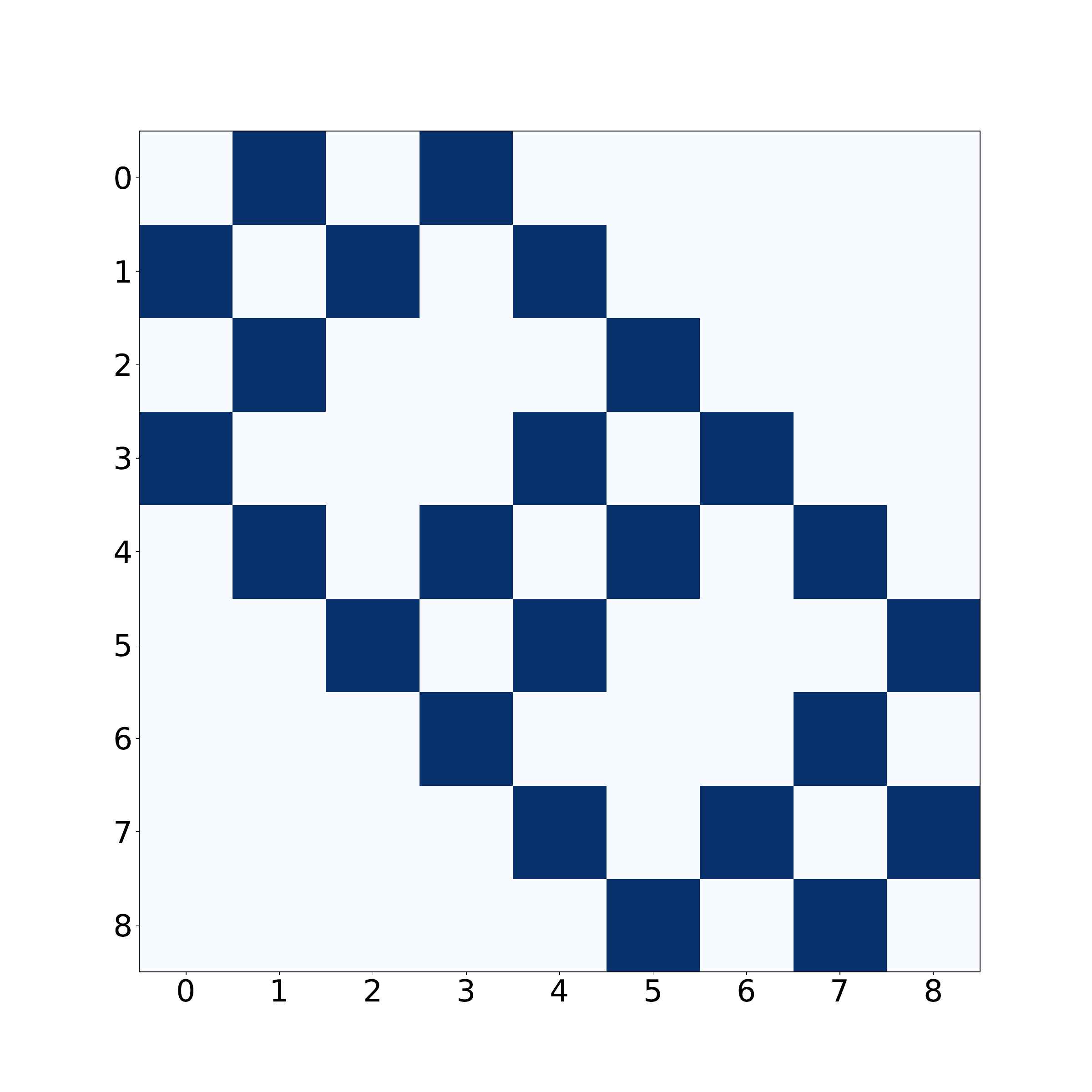}}
    \subfloat[][Our estimate $\hat{L}_{\text{sym}}$]{\includegraphics[width=0.25\textwidth]{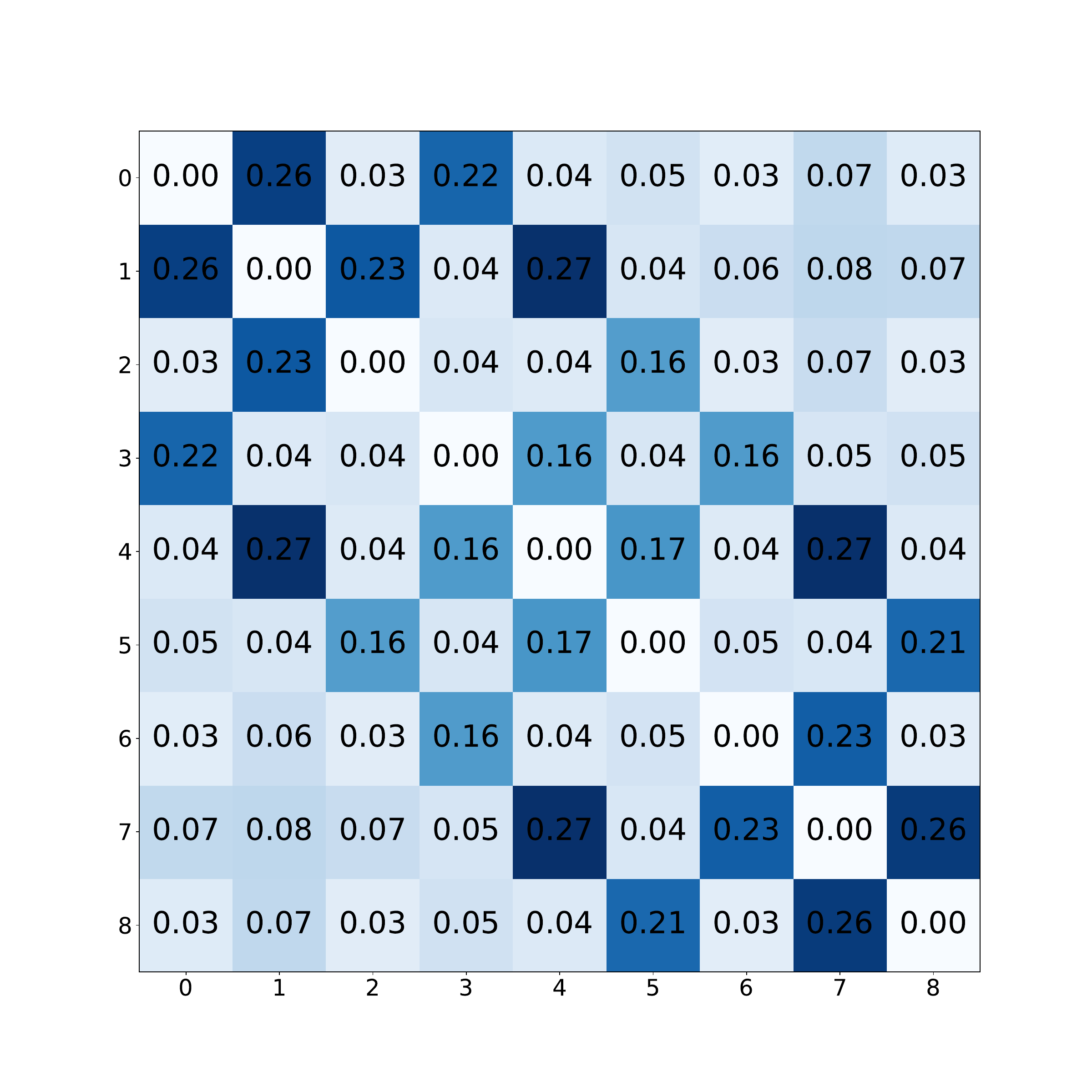}}
    \subfloat[][Transfer entropy matrix]{\includegraphics[width=0.25\textwidth]{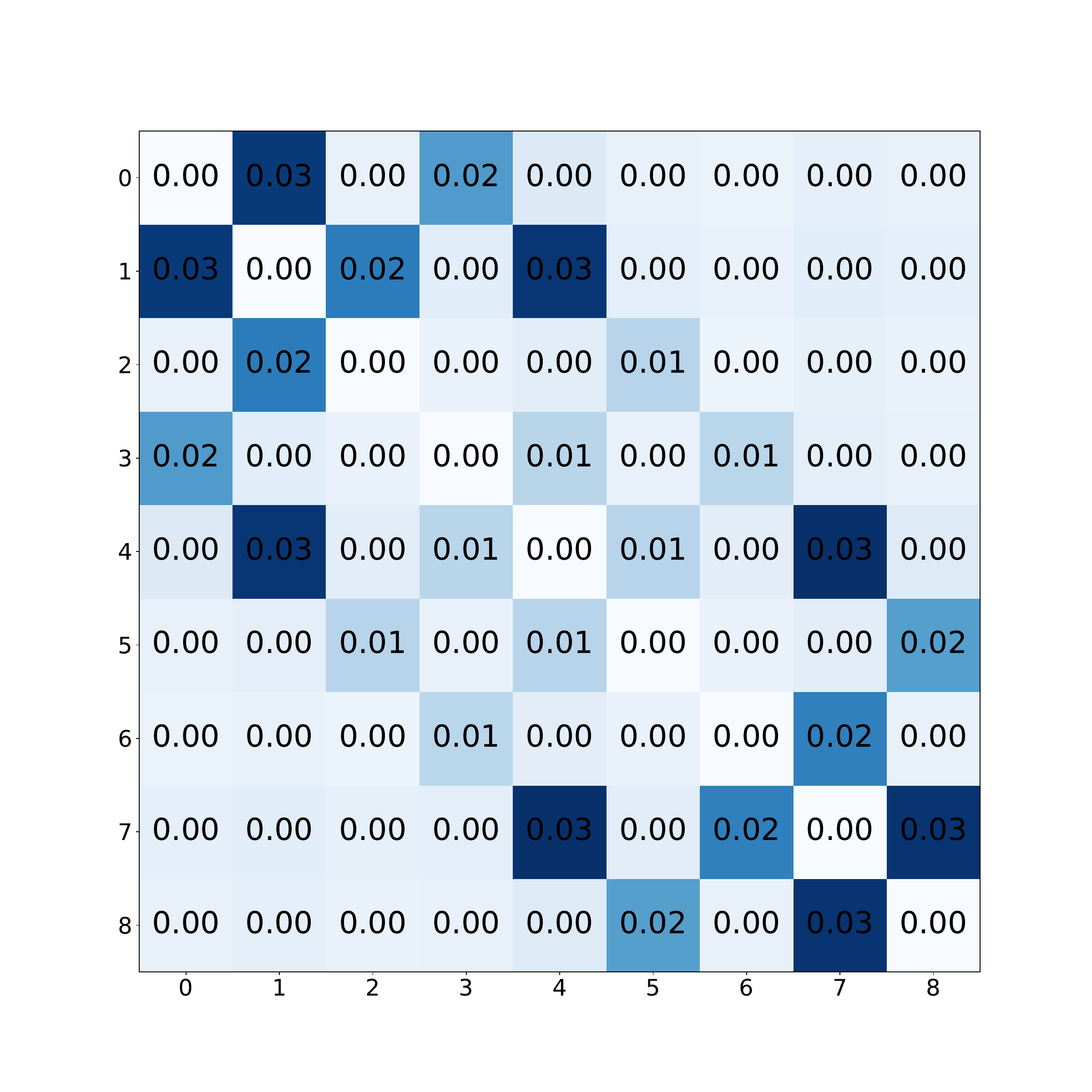}}
    \subfloat[][Varying transmission length]{\includegraphics[width=0.25\textwidth]{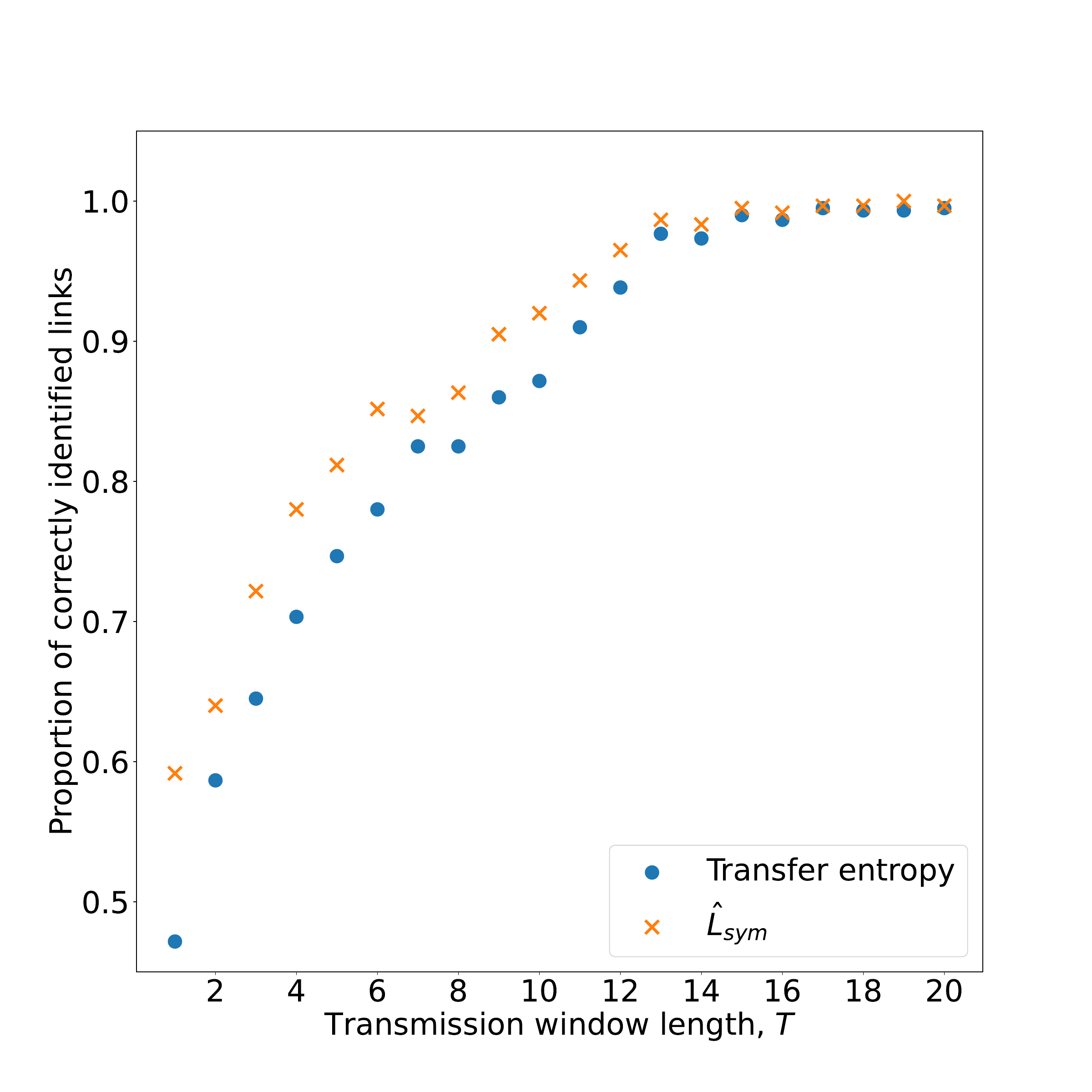}}
    \vspace{0.5cm}
    \caption{
    Experimental results for a network topology where nodes are arranged in a three-by-three grid.  See \cref{fig:cycle} for a description. 
   }
    \label{fig:3by3grid}
\end{figure*}

\begin{figure*}[h!]
    \centering
    \subfloat[][Adjacency matrix]{\includegraphics[width=0.25\textwidth]{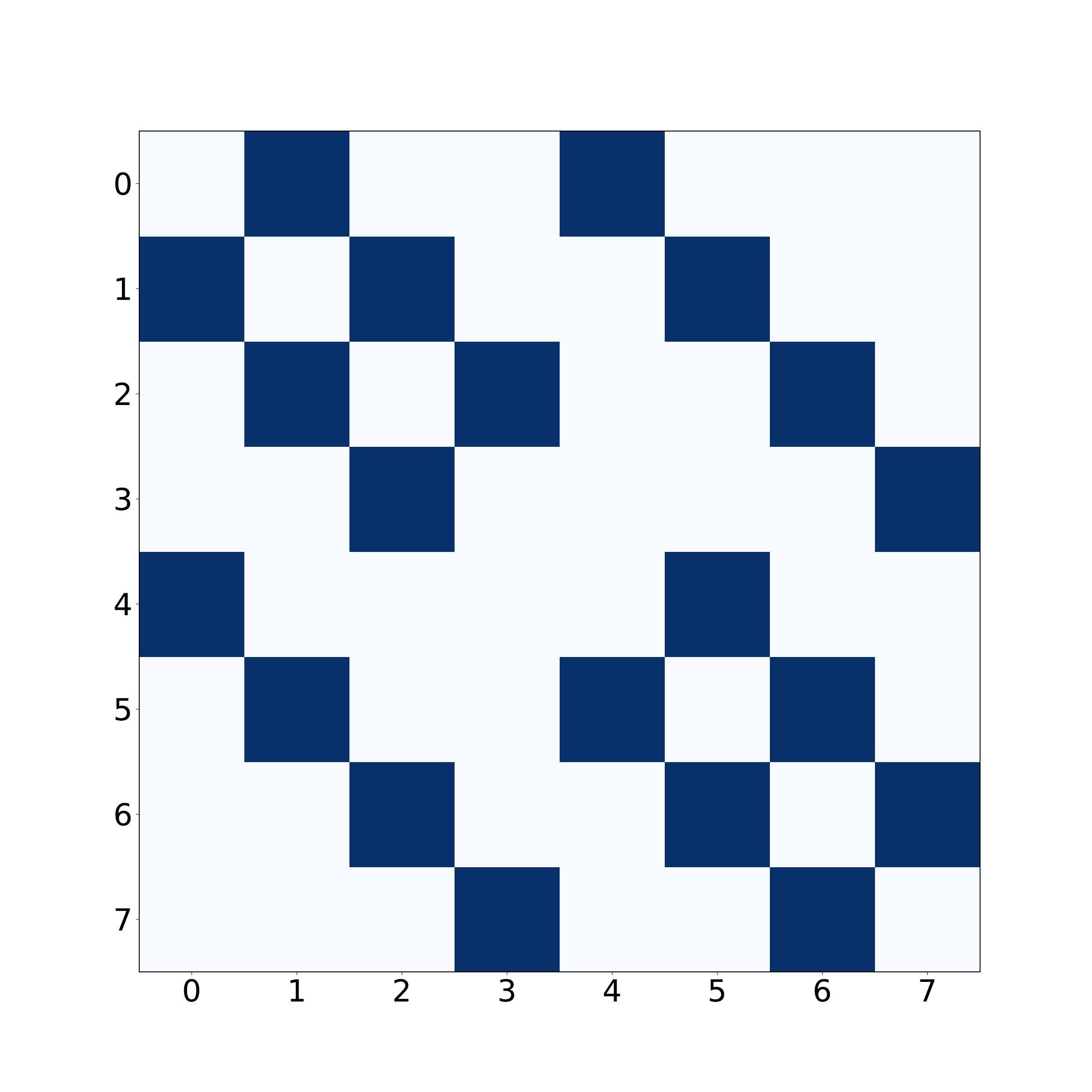}}
    \subfloat[][Our estimate $\hat{L}_{\text{sym}}$]{\includegraphics[width=0.25\textwidth]{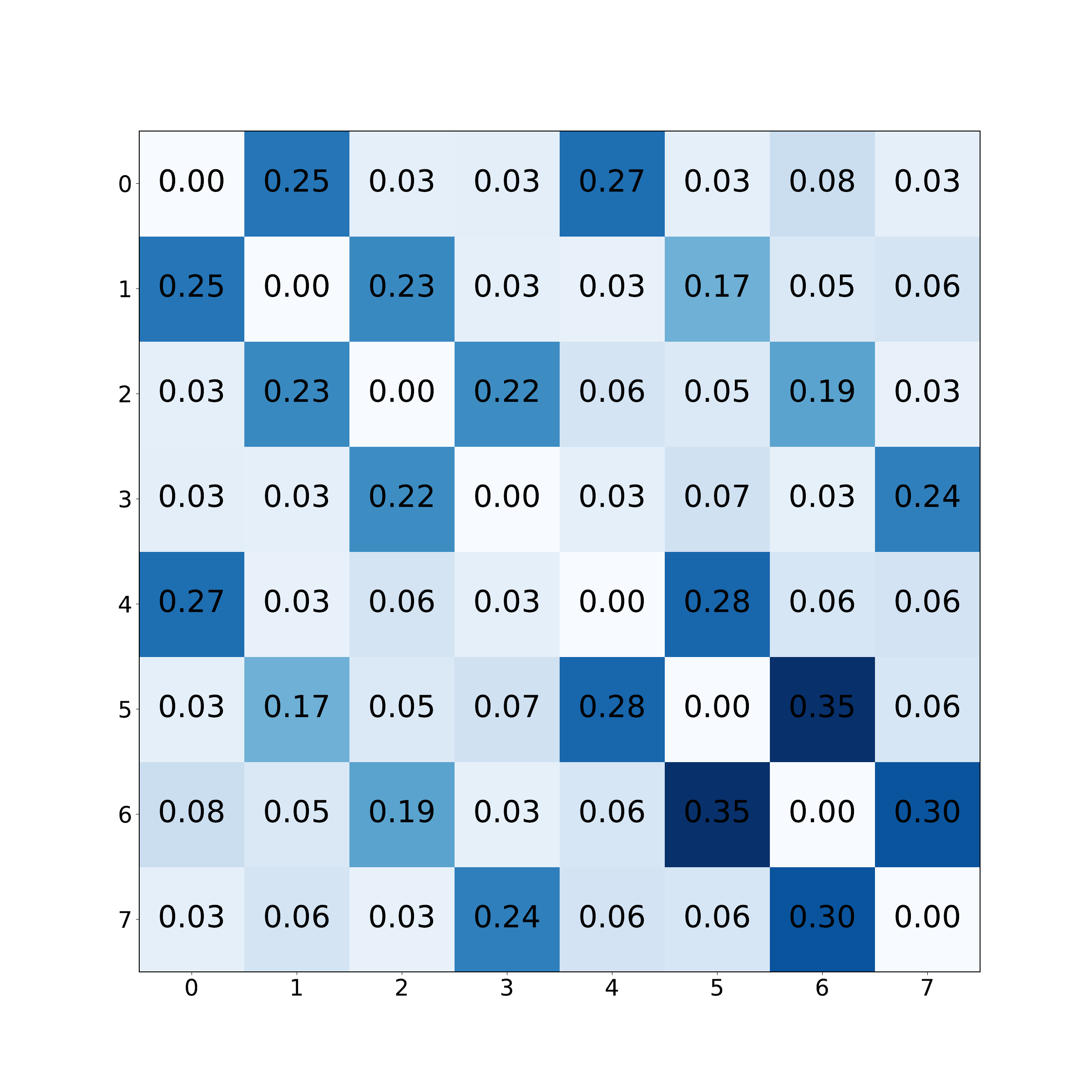}}
    \subfloat[][Transfer entropy matrix]{\includegraphics[width=0.25\textwidth]{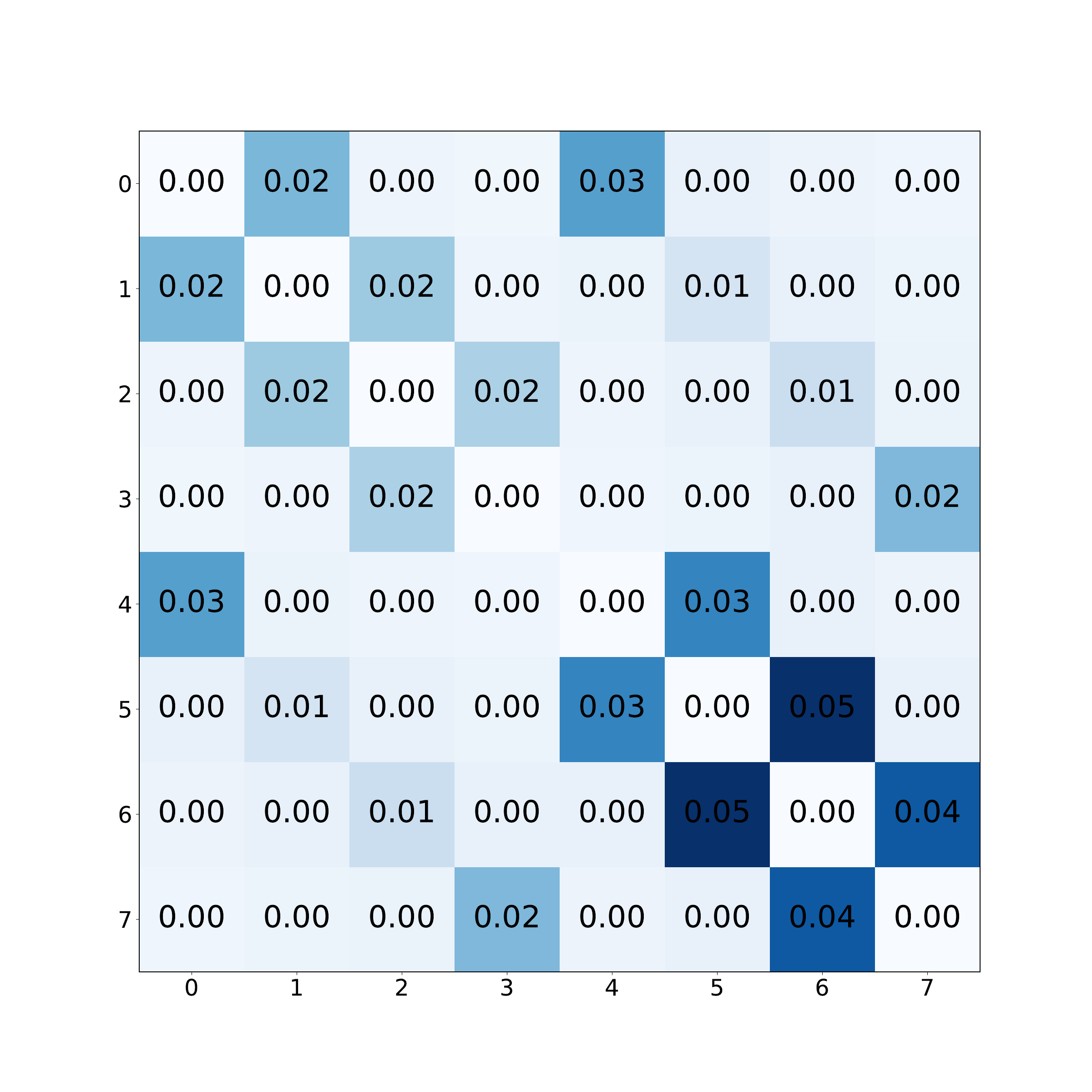}}
    \subfloat[][Varying transmission length]{\includegraphics[width=0.25\textwidth]{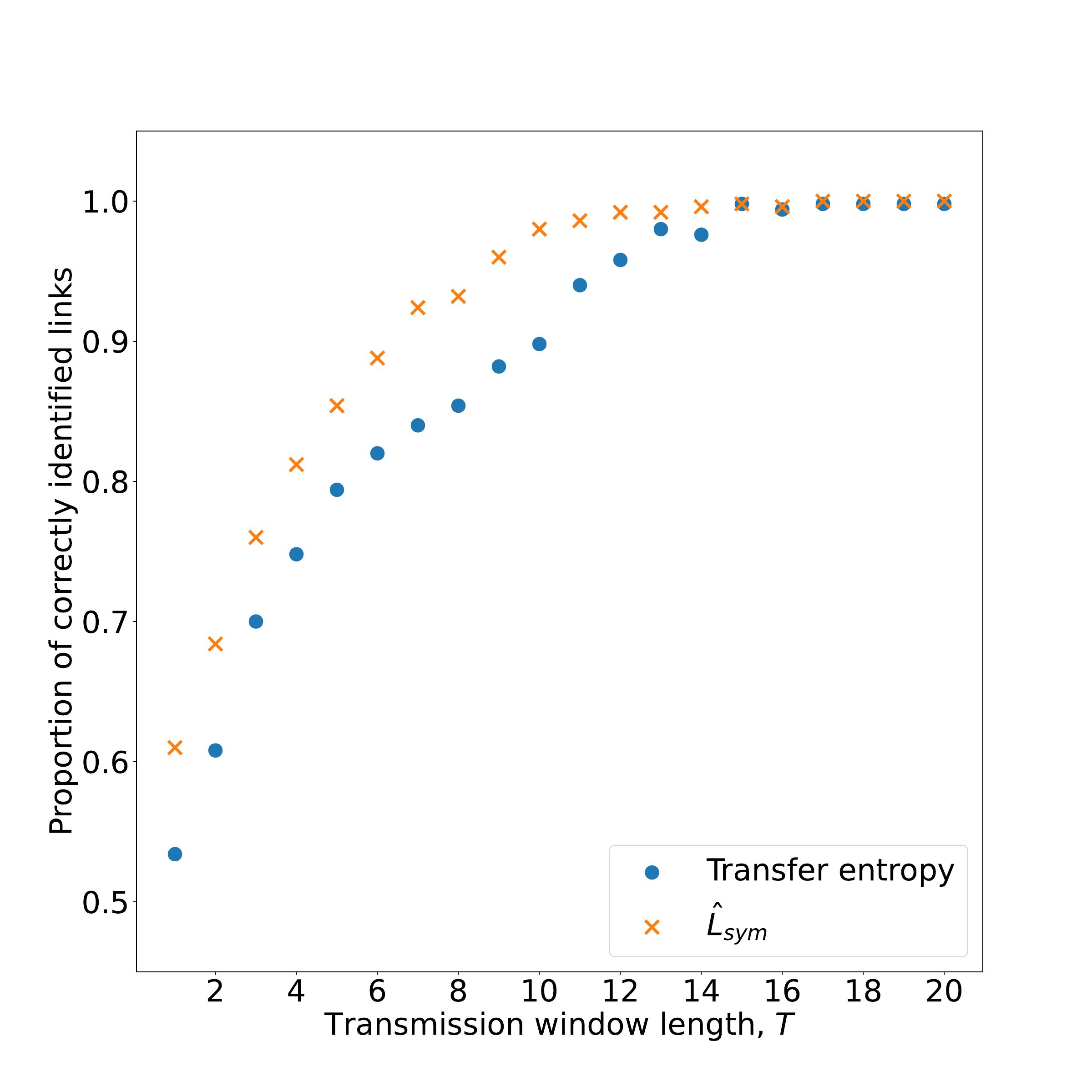}}
    \vspace{0.5cm}
    \caption{
    Experimental results for a network topology where nodes are arranged in a four-by-two grid.  See \cref{fig:cycle} for a description.    }
    \label{fig:ladder}
\end{figure*}

\begin{figure*}[h!]
    \centering
    \subfloat[][Adjacency matrix]{\includegraphics[width=0.25\textwidth]{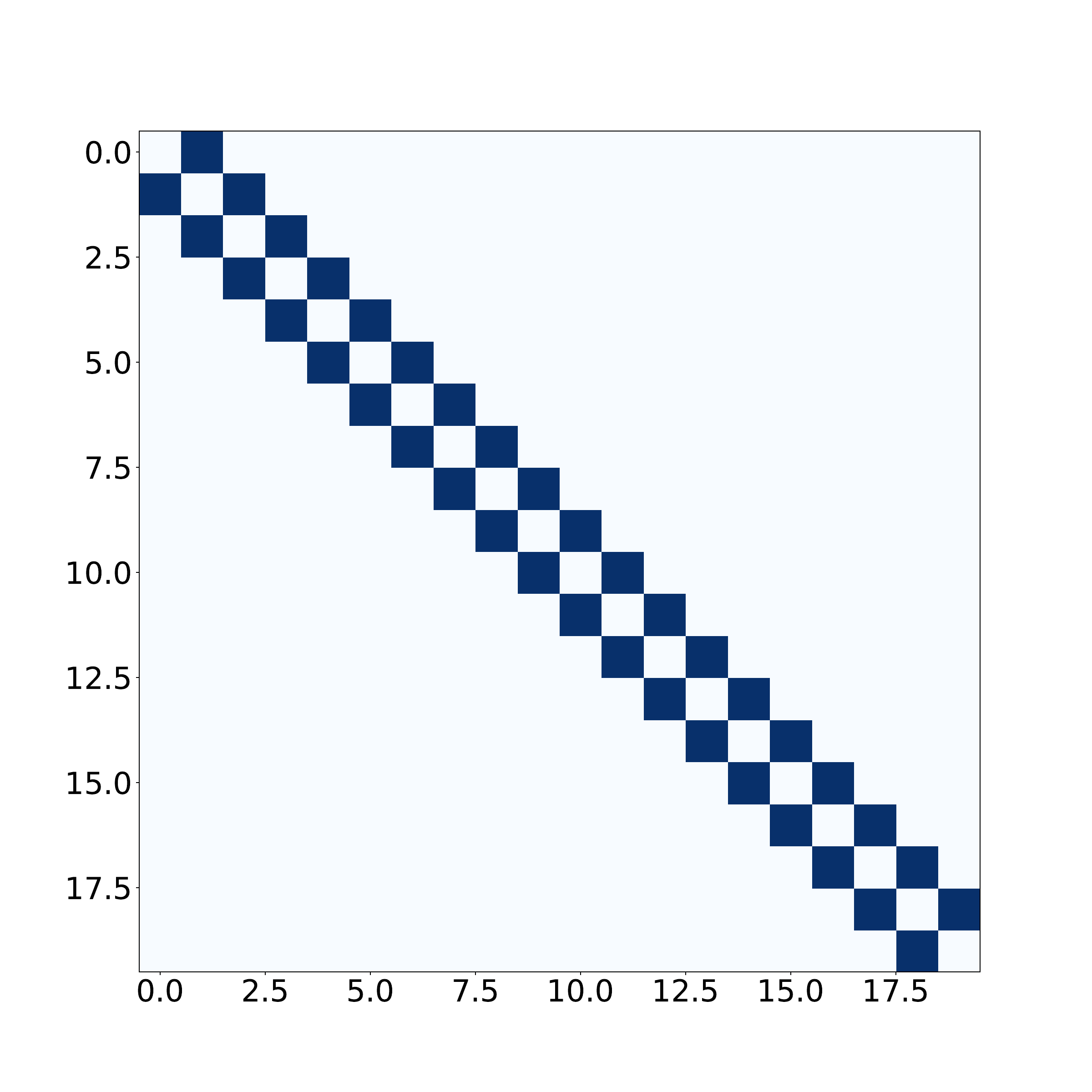}}
    \subfloat[][Our estimate $\hat{L}_{\text{sym}}$]{\includegraphics[width=0.25\textwidth]{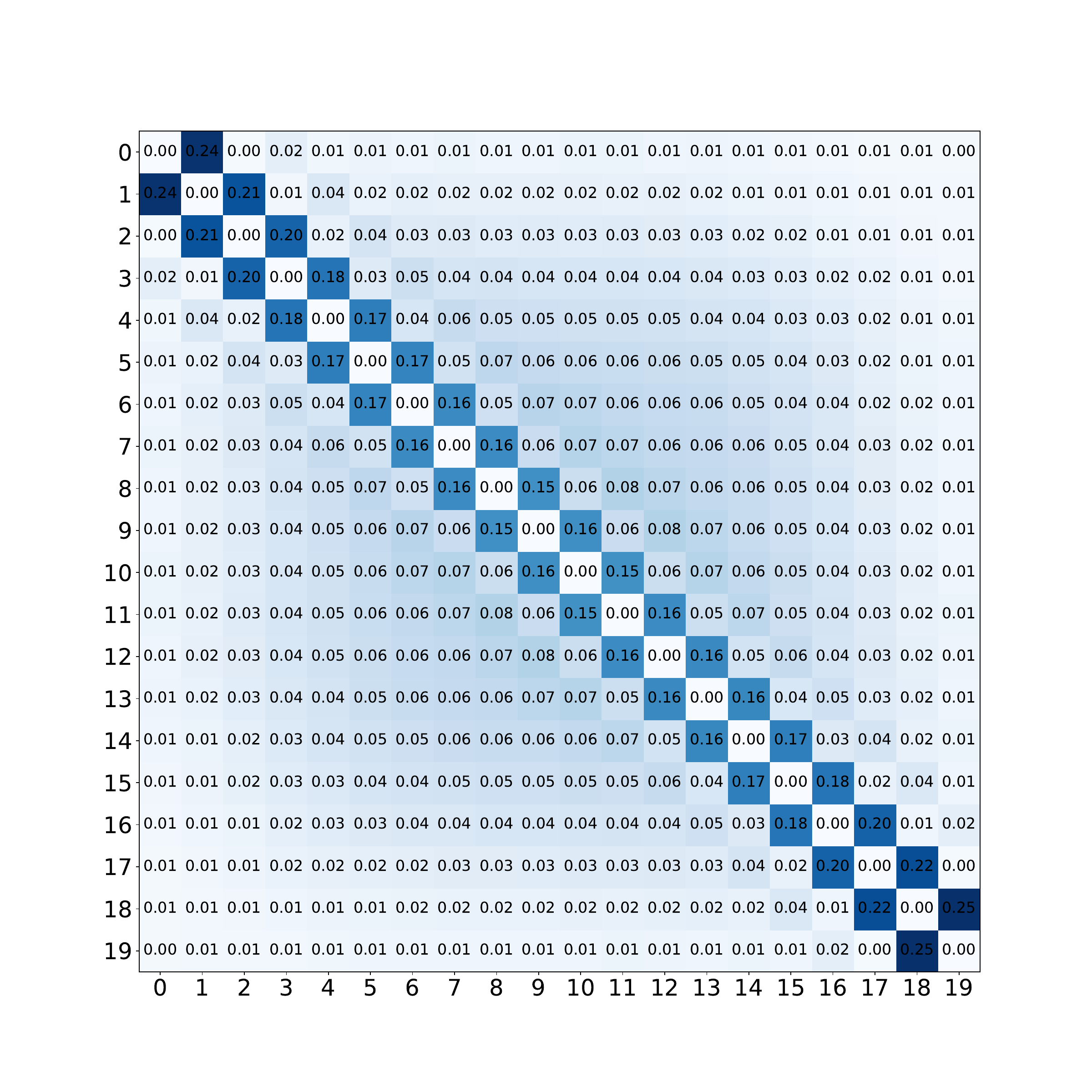}}
    \subfloat[][Transfer entropy matrix]{\includegraphics[width=0.25\textwidth]{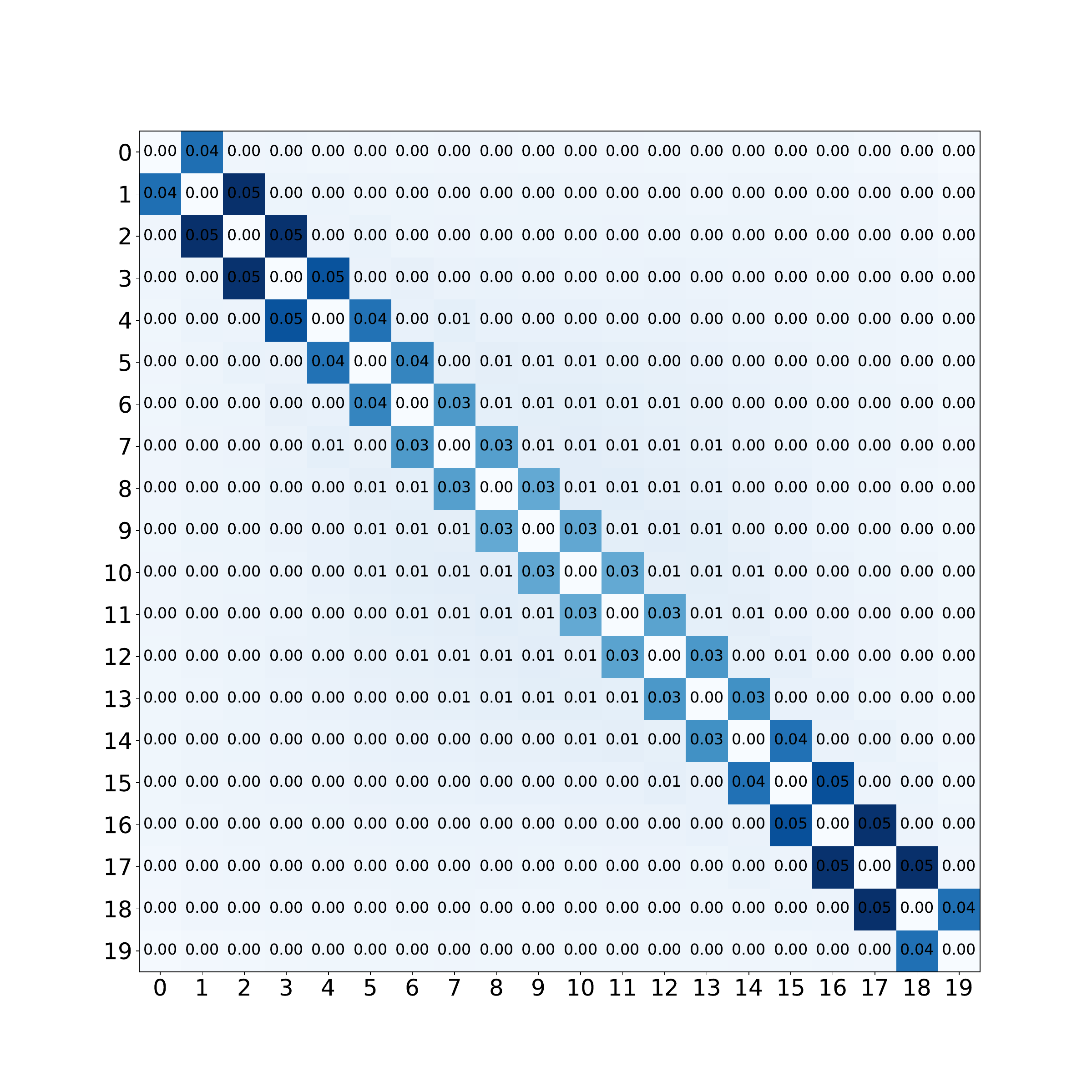}}
    \subfloat[][Varying transmission length]{\includegraphics[width=0.25\textwidth]{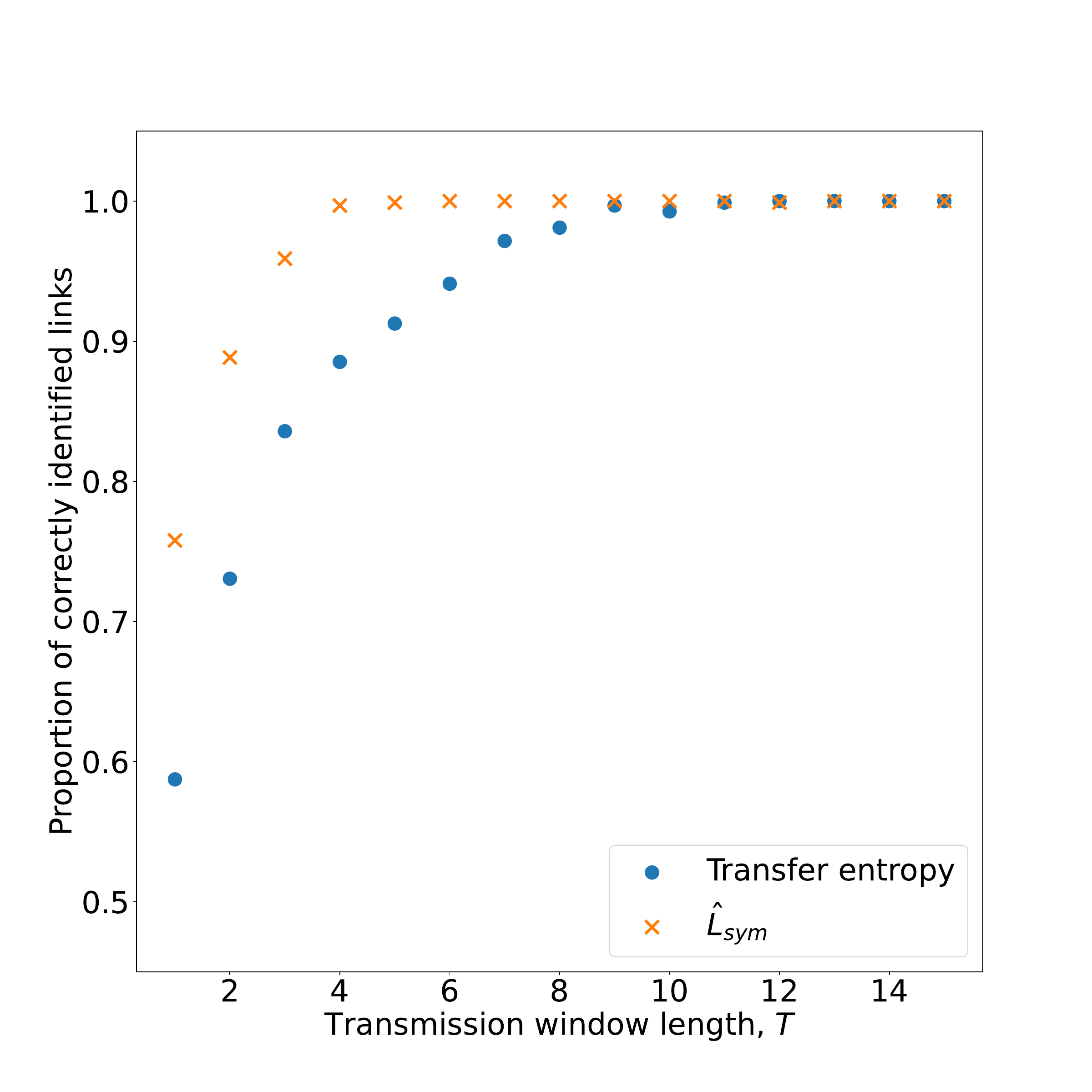}}
    \vspace{0.5cm}
    \caption{
    Experimental results for a network topology of twenty nodes arranged in a path.  See \cref{fig:cycle} for a description.    }
    \label{fig:path20}
\end{figure*}

%

\subsection{A more realistic setting}
\label{sec:real}
We now consider a more realistic scenario in which we do not have perfect information about the time series matrix TS. For simplicity, we will assume to know the number of nodes in the network, but not their position. To obtain information about packet transmission, we assume we are able to deploy a set of $s$ RF sensors in the environment. When a packet is transmitted from a node $u$, RF sensors positioned close to $u$ will be able to detect the corresponding power burst. We use the \emph{Log Distance} path loss model in our \texttt{ns-3} simulations: this implies that, the farther away from $u$ a sensor $\sigma$ is, the more attenuated the received signal by the sensor.

We again divide the time into small intervals $\{1,\dots,T'\}$ (we assume a power burst is detected by at least a sensor at each interval, and discard intervals in which no power burst is detected); we can then construct a matrix $PD \in \mathbb{R}^{m \times T'}$, where $PD(\sigma,t)$ represents the power detected at time $t$ by sensor $\sigma$. Suppose that, for a pair of time intervals $t,t'$,  the two vectors $PD(\cdot,t)$ and $PD(\cdot,t' )$ are somewhat similar: this indicates that the same node might be responsible for the detected power burst. This observation suggests a natural way to discern if the same node is transmitting during different time intervals: partition the vectors $\{PD(\cdot,t)\}_{t=1}^{T'}$ using a clustering technique such as $k$-means. Time intervals assigned to the same cluster corresponds to power bursts we associate to the same node. In this way it is possible to construct a matrix which, hopefully, approximates the time series matrix TS (since we do not assume any information about the identity of the nodes, such as their position, this is only possible up to a relabelling of the nodes).

We perform experiments using the same \texttt{ns-3} experimental setup discussed at beginning of the section, but instead of fixing the network topology we place $n$ transmitters and $s$ RF sensors at random in a rectangular box. For example, in \cref{fig:box} we have randomly placed $25$ transmitters and $100$ sensors in a square box of side $120m$. Only transmitters placed at most $50m$ apart are connected by a link in the network. This gives rise to a network topology with adjacency matrix displayed in \cref{fig:box}(d). We simulate the network behaviour with \texttt{ns-3} exchanging $500$ messages between randomly chosen pair of nodes. As explained above, we use $k$-means clustering (with $k = n$) to group together power bursts measured by different sensors in the same interval. Since we do not have any information about the identity of the transmitters (recall, we do not assume any information about the location of the sensors or transmitters), to evaluate how well measured power bursts are clustered together, we label each cluster computed by $k$-means according to the transmitter that corresponds to the largest number of bursts in the cluster. We display in \cref{fig:box}(b) the percentage of power bursts detected that are associated to the correct transmitter for different lengths of the transmission window. The results are averaged over $20$ random realisations in which we fix the same network topology but randomly vary the position of the sensors. As expected, a larger transmission window results in a higher accuracy thanks to a lower number of concurrent transmissions. 

We then compute our estimator $\hat{L}_{\text{sym}}$ together with the transfer entropy matrix $M_{\text{TE}}$. Examples of these two matrices for a particular random realisation are shown in \cref{fig:box} (e) and (f). Notice how these two matrices appear quite similar. In \cref{fig:box}(c) we compare the performances of our method with transfer entropy. Again, we measured the number of times the top $m$ entries in the two matrices correspond to links in the network, where $m$ is the number of links in the network. Both methodologies appear to perform reasonably well for $T\ge 10$, with our estimation procedure slightly but consistently improving upon transfer entropy.

We repeat the same experiments with a bigger network: we randomly place $50$ transmitters and $150$ sensors in a rectangular box of $340m \times 200m$. We simulate the network with \texttt{ns-3} exchanging now $2000$ messages between randomly chosen pairs. We perform $5$ random realisations, randomly placing both transmitters and sensors. Experimental results are shown in \cref{fig:bigbox}. We can observe a degradation of the performances, both in the proportion of identified transmitters (see \cref{fig:bigbox}(b)) and in the proportion of links of the network that corresponds to the top $m$ entries in both transfer entropy and $\hat{L}_{\text{sym}}$ matrix (see \cref{fig:bigbox}(c)). The  first observation suggests that, in large networks, a simple application of $k$-means might not be enough to identify the correct transmitter for most power bursts; future work might want to exploit potential information about the position of sensors or different source separation procedures such as using independent component analysis \cite{TG20}. The decrease in performance in correctly identifying links is not unexpected: the number of possible links increases quadratically in the number of transmitters while the number of actual links increases, roughly, linearly. Nevertheless, to help us better understand the reasons behind this decrease in performance, let us study closely \cref{fig:bigbox} (f), which compares the top entries of $\hat{L}_{\text{sym}}$ against the actual links in the network: first of all, notice how the main ``backbone'' of the network is essentially preserved. As suggested by our theoretical analysis, our method does not necessarily preserve individual links but rather structural properties of the network. Furthermore, false positive links typically connect nodes that are only two hops apart in the network. This is not surprising since if there exist links $\{u,v\}$ and $\{v,w\}$, the time series corresponding to $u$ and $w$ will be correlated. This suggests it might be possible to reduce this error by, for example, combining information about time series with different interval lengths. We leave this as future work.

\begin{figure*}[h!]
    \centering
    \subfloat[][Placement of transmitters and sensors]{\includegraphics[width=0.3\textwidth]{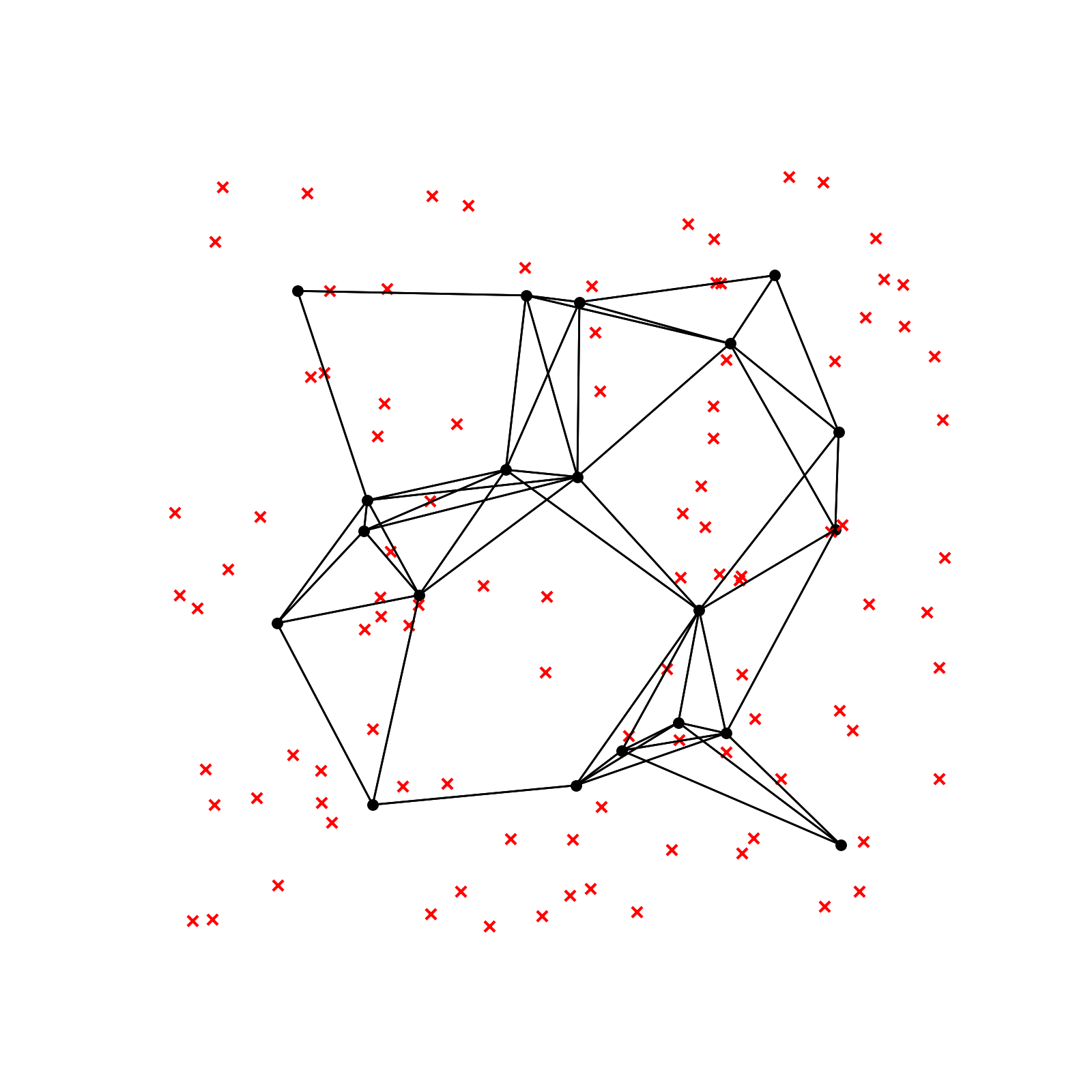}}
    \subfloat[][Correctly identified transmitters]{\includegraphics[width=0.3\textwidth]{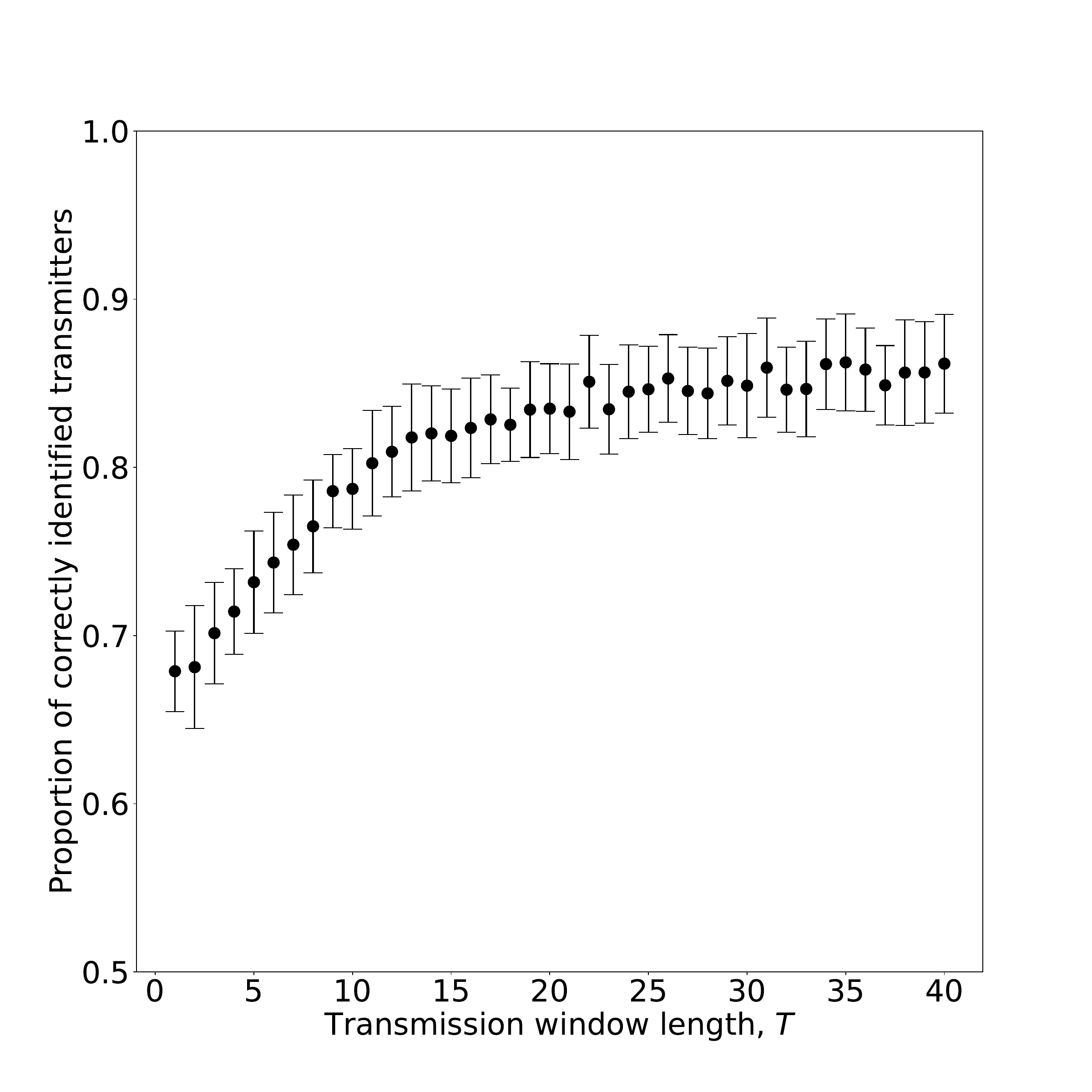}}
    \subfloat[][Correctly identified links]{\includegraphics[width=0.3\textwidth]{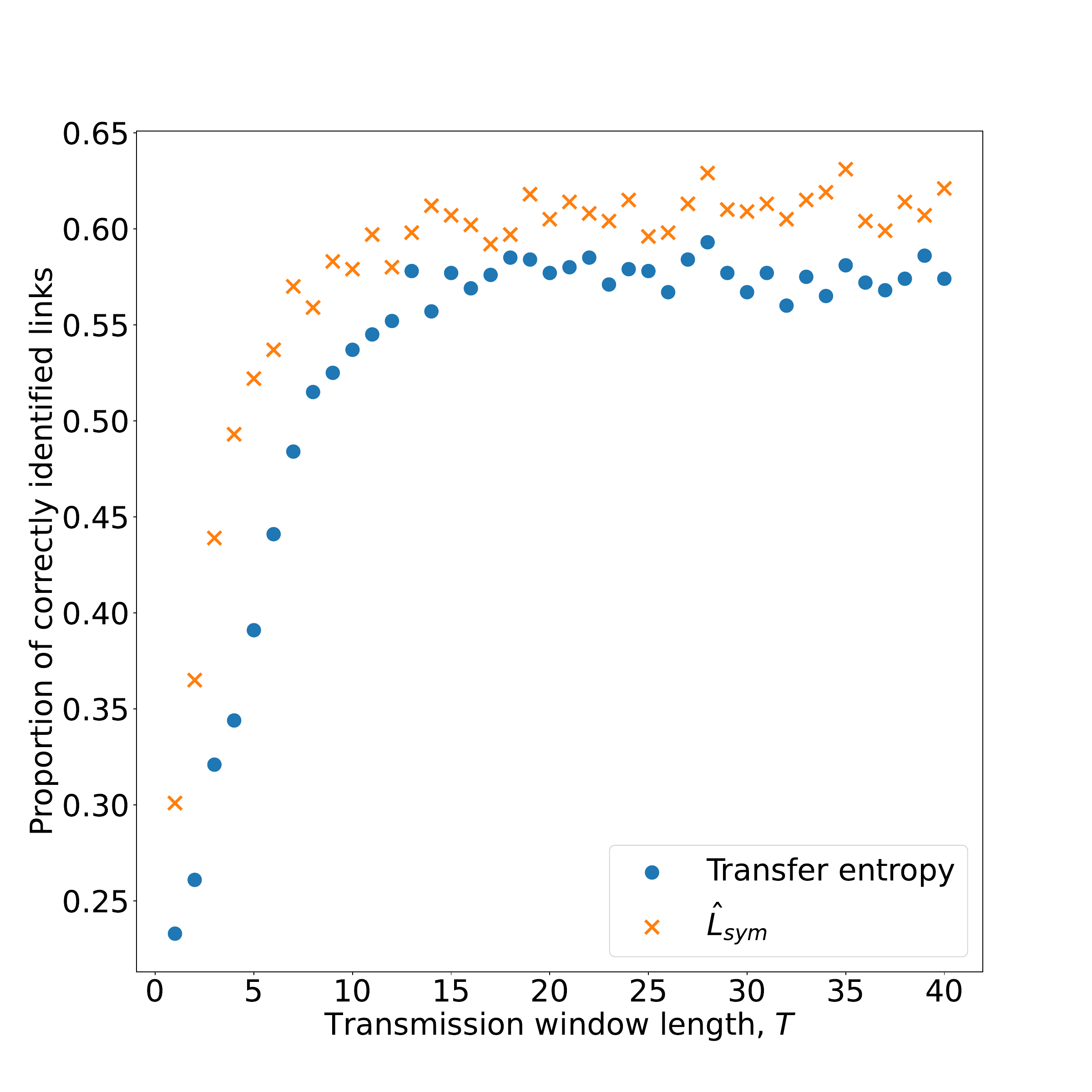}} \\
    \subfloat[][Adjacency matrix]{\includegraphics[width=0.3\textwidth]{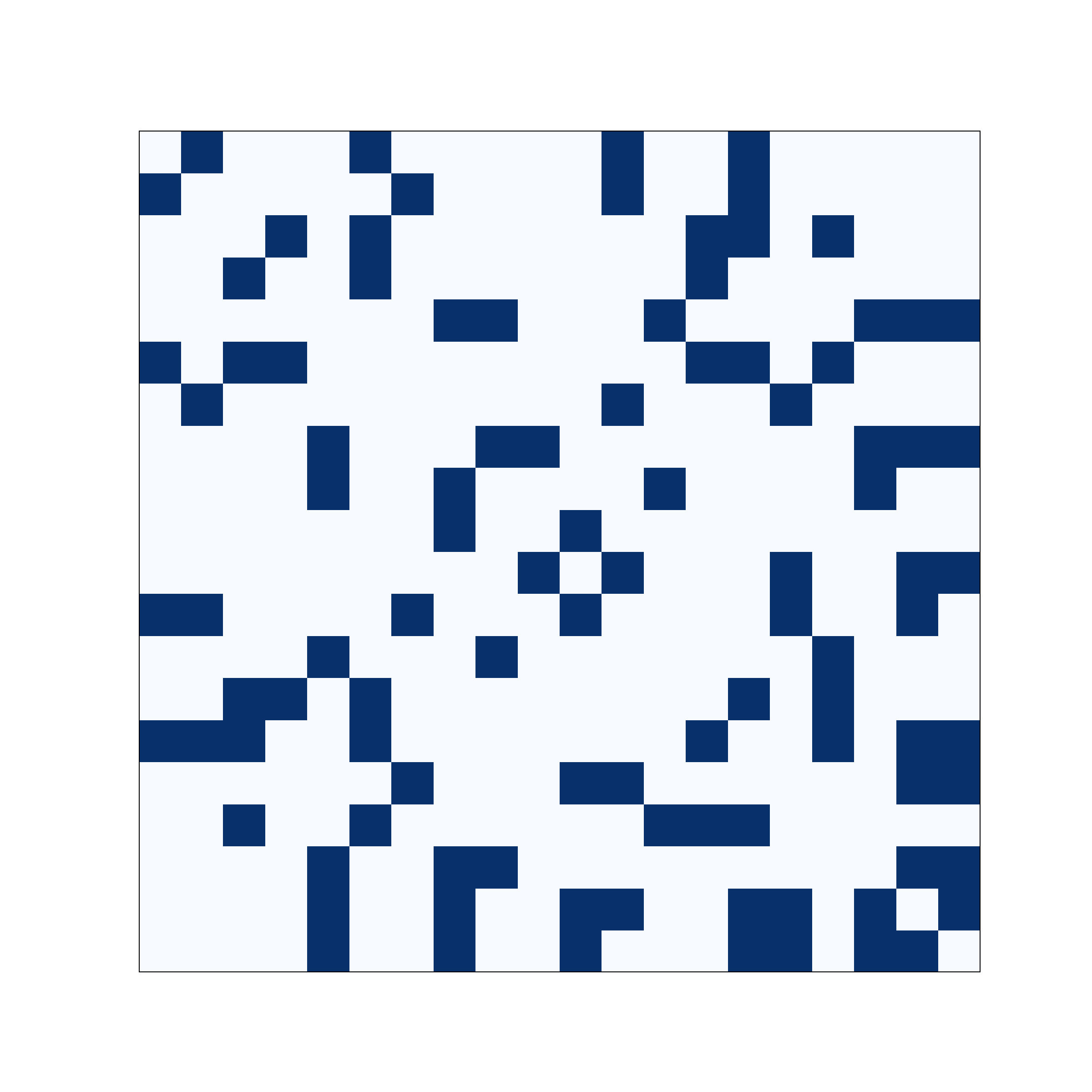}}
    \subfloat[][Our estimate $\hat{L}_{\text{sym}}$]{\includegraphics[width=0.3\textwidth]{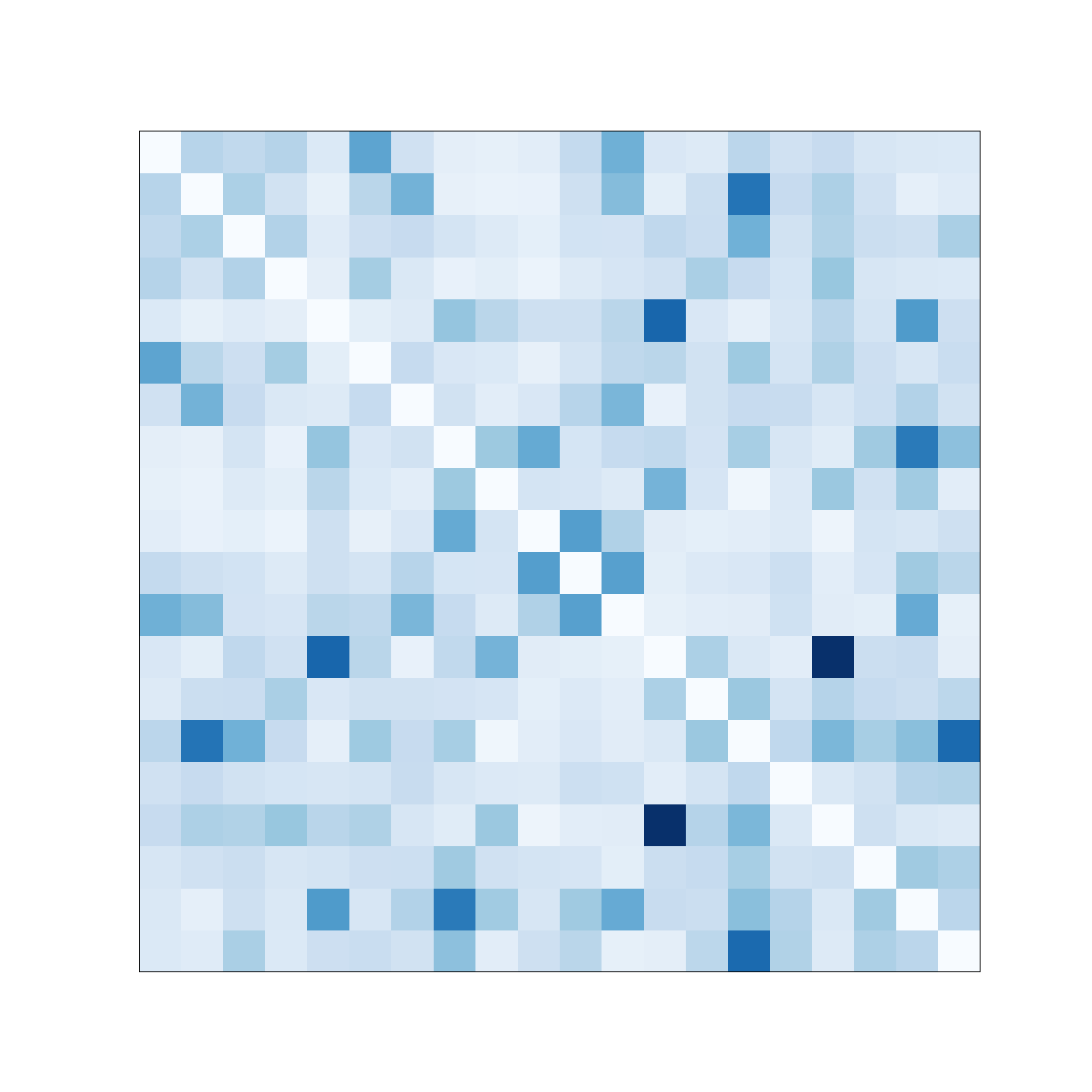}}
    \subfloat[][Transfer entropy matrix]{\includegraphics[width=0.3\textwidth]{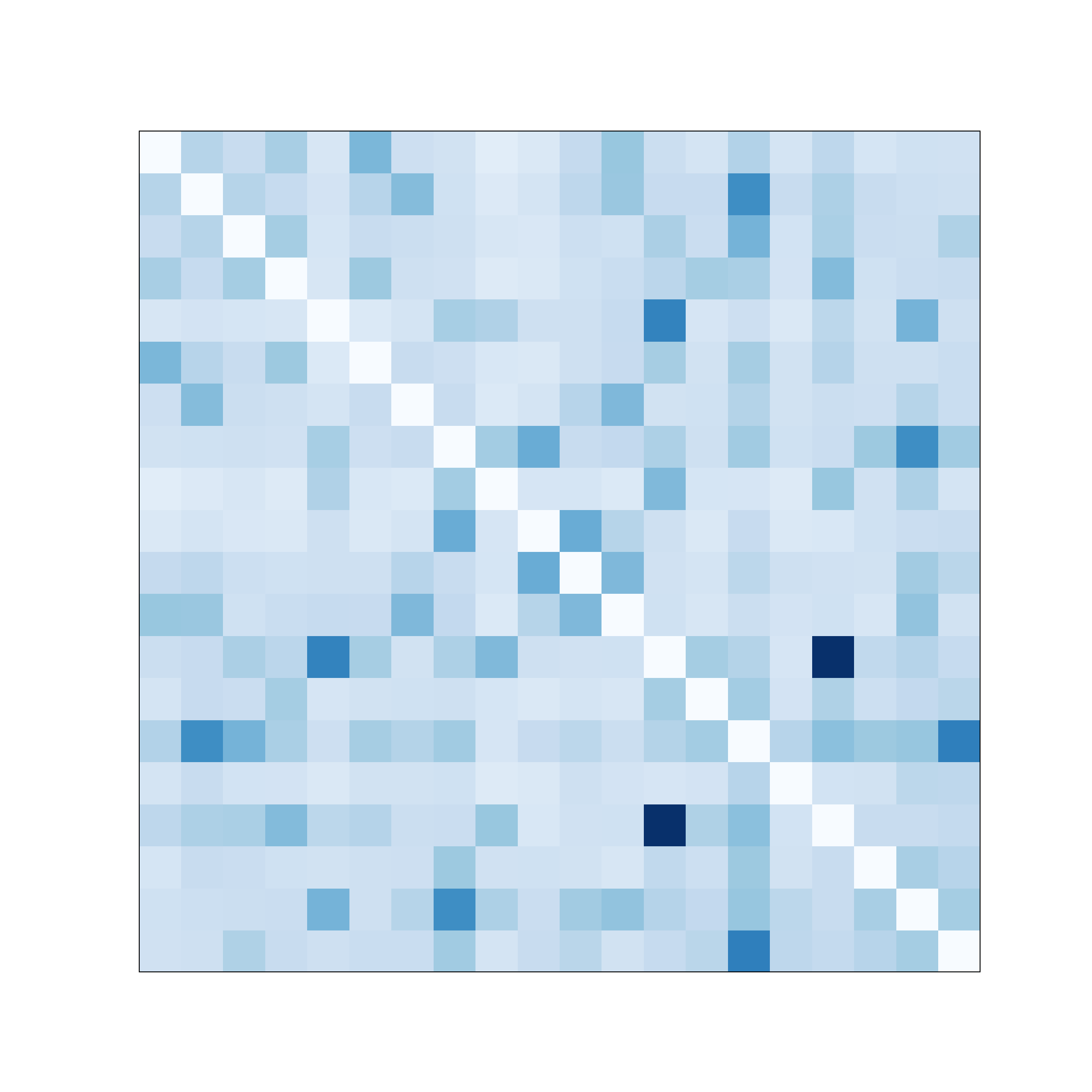}} 
    \vspace{0.5cm}
    \caption{
    Experimental results for a random network constructed by placing $25$ transmitters and $100$ RF sensors at random in a $120m \times 120m$ square box. Figure (a) displays the position of the transmitters (black circles) and of the sensors (red crosses) for a particular random realisation, together with the resulting network topology. Figure (b) shows the number of power bursts correctly associated to the corresponding transmitters for a varying transmission window length (we display the average number over 20 random realisations, together with a 95\% confidence interval). Figure (c) represents the proportion of correctly identified links by our methodology and transfer entropy, for a varying length of the transmission window (averaged over  20 random realisations).  Figure (d) displays the adjacency matrix of the network. Figure (e) the estimated matrix $\hat{L}_{\text{sym}}$ for a transmission window length $T=10s$. Figure (f) the  transfer entropy matrix ($T=10s$).
    }
    \label{fig:box}
\end{figure*}

\begin{figure*}[h!]
    \centering
    \subfloat[][Adjacency matrix]{\includegraphics[width=0.3\textwidth]{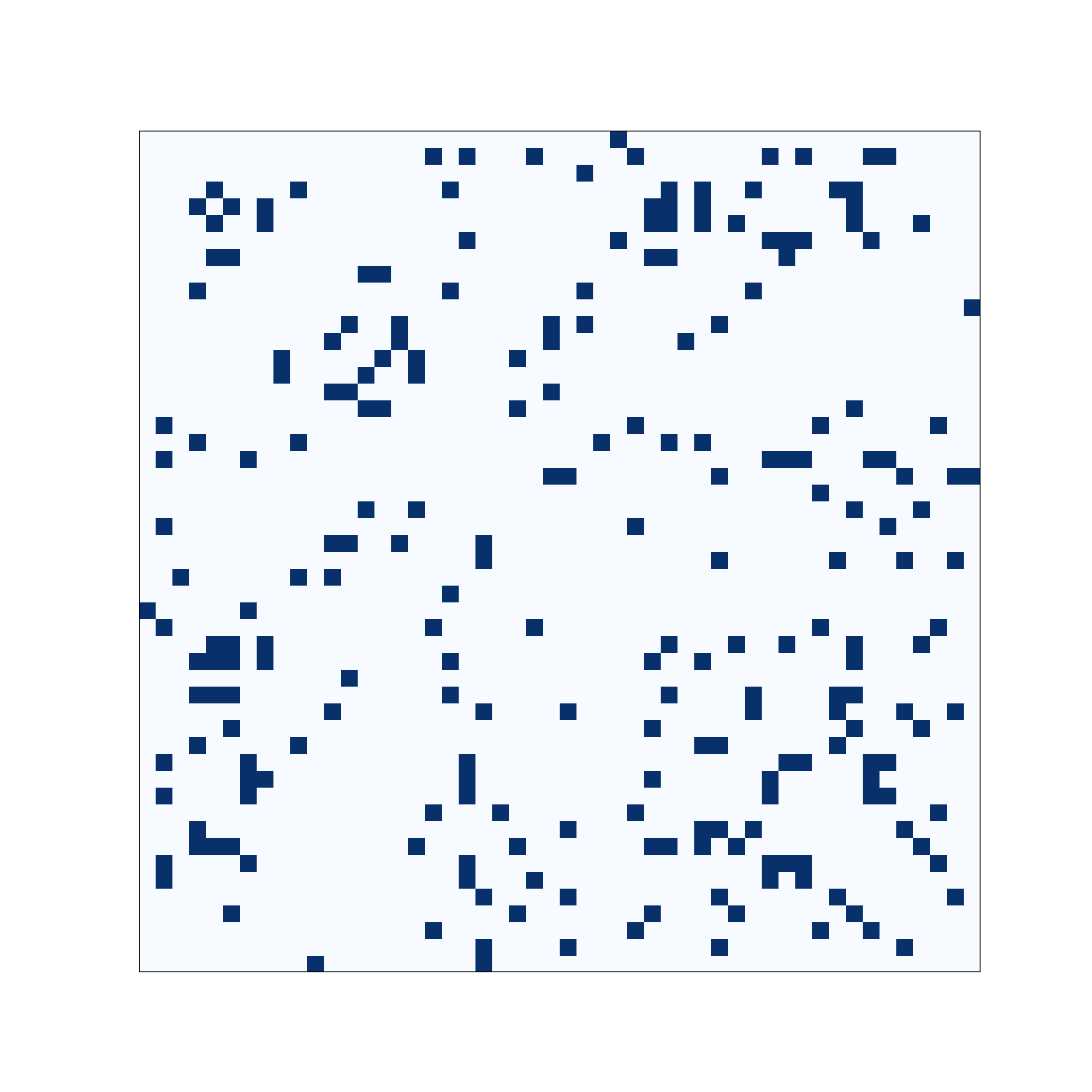}}
    \subfloat[][Our estimate $\hat{L}_{\text{sym}}$]{\includegraphics[width=0.3\textwidth]{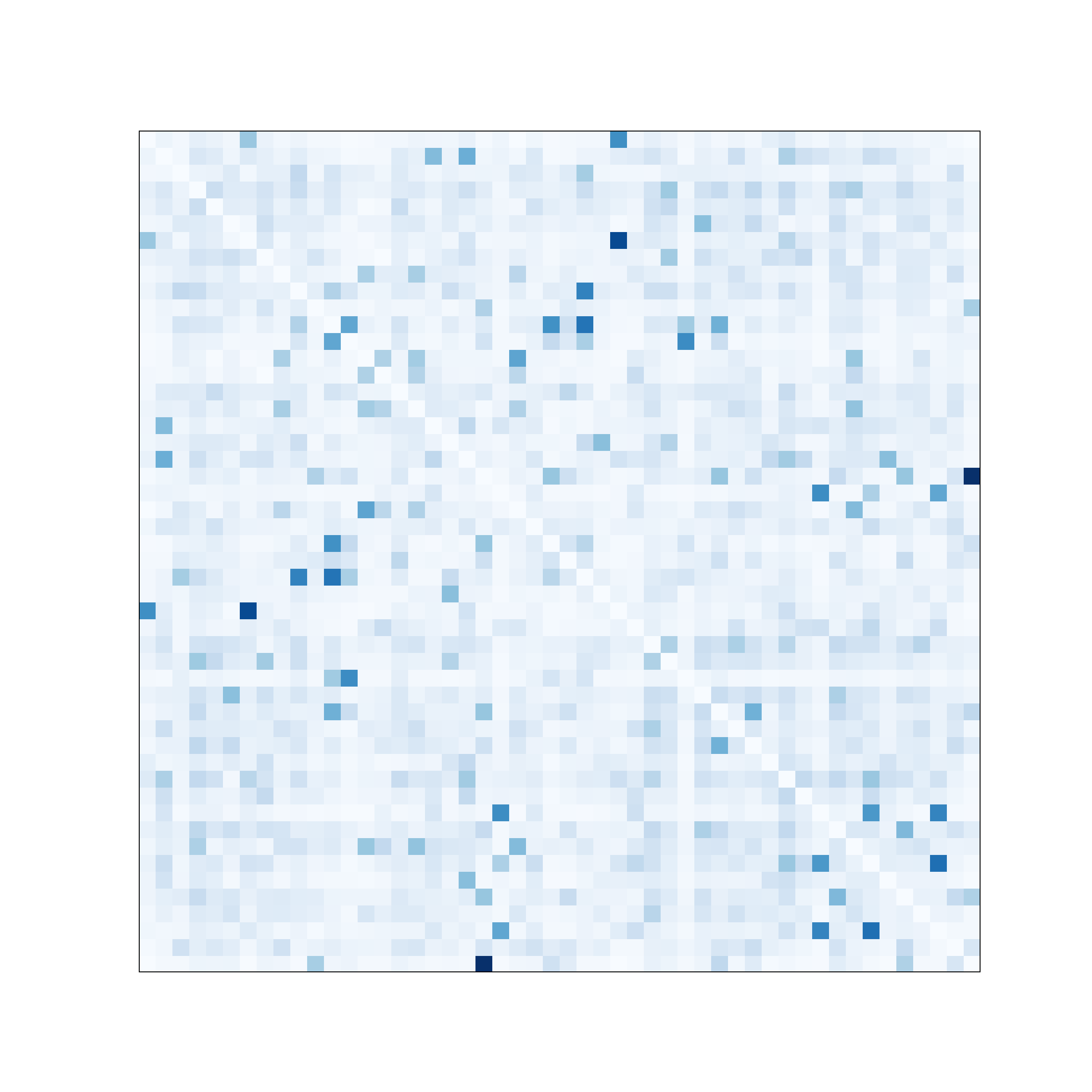}}
    \subfloat[][Transfer entropy matrix]{\includegraphics[width=0.3\textwidth]{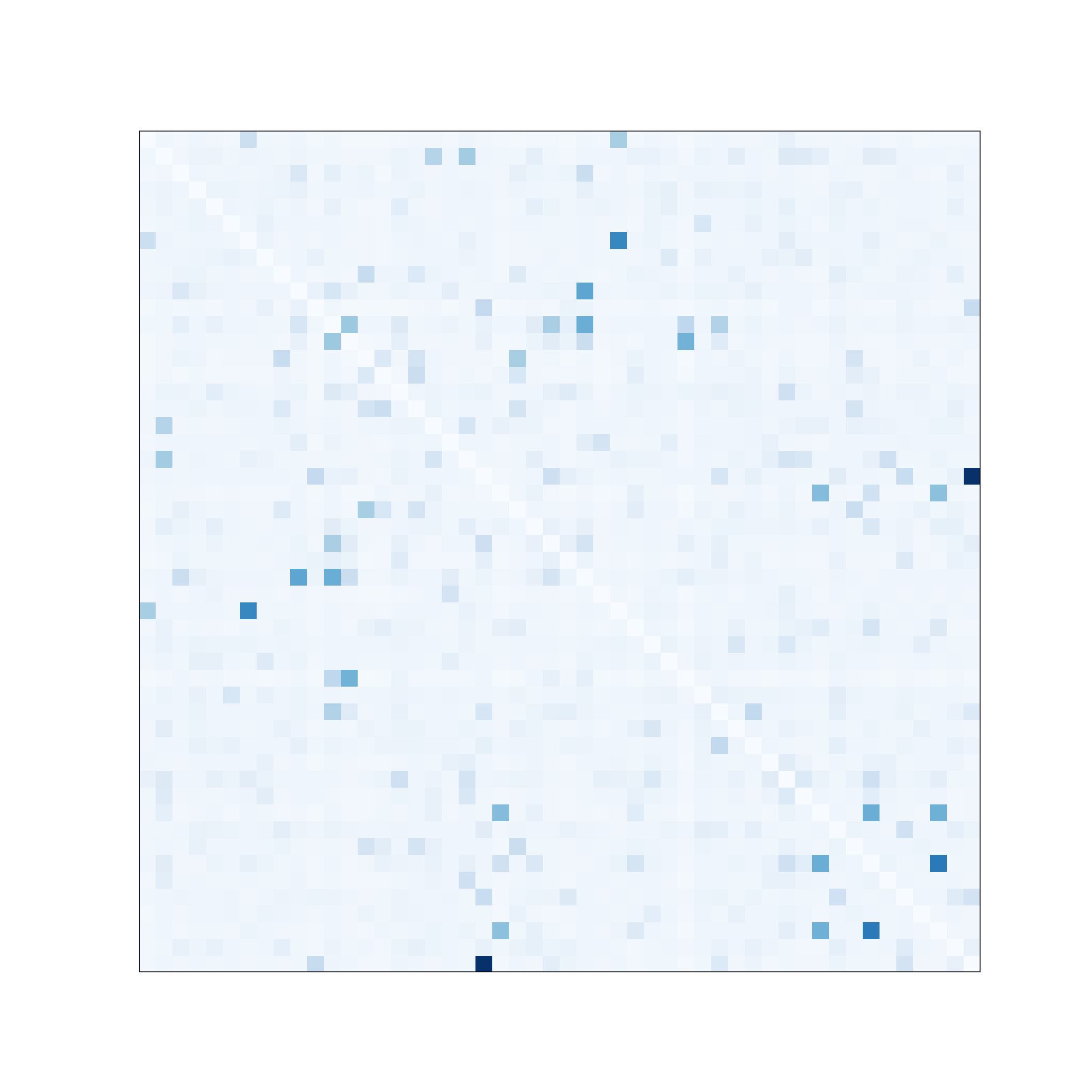}} \\
    \subfloat[][Correctly identified transmitters]{\includegraphics[width=0.3\textwidth]{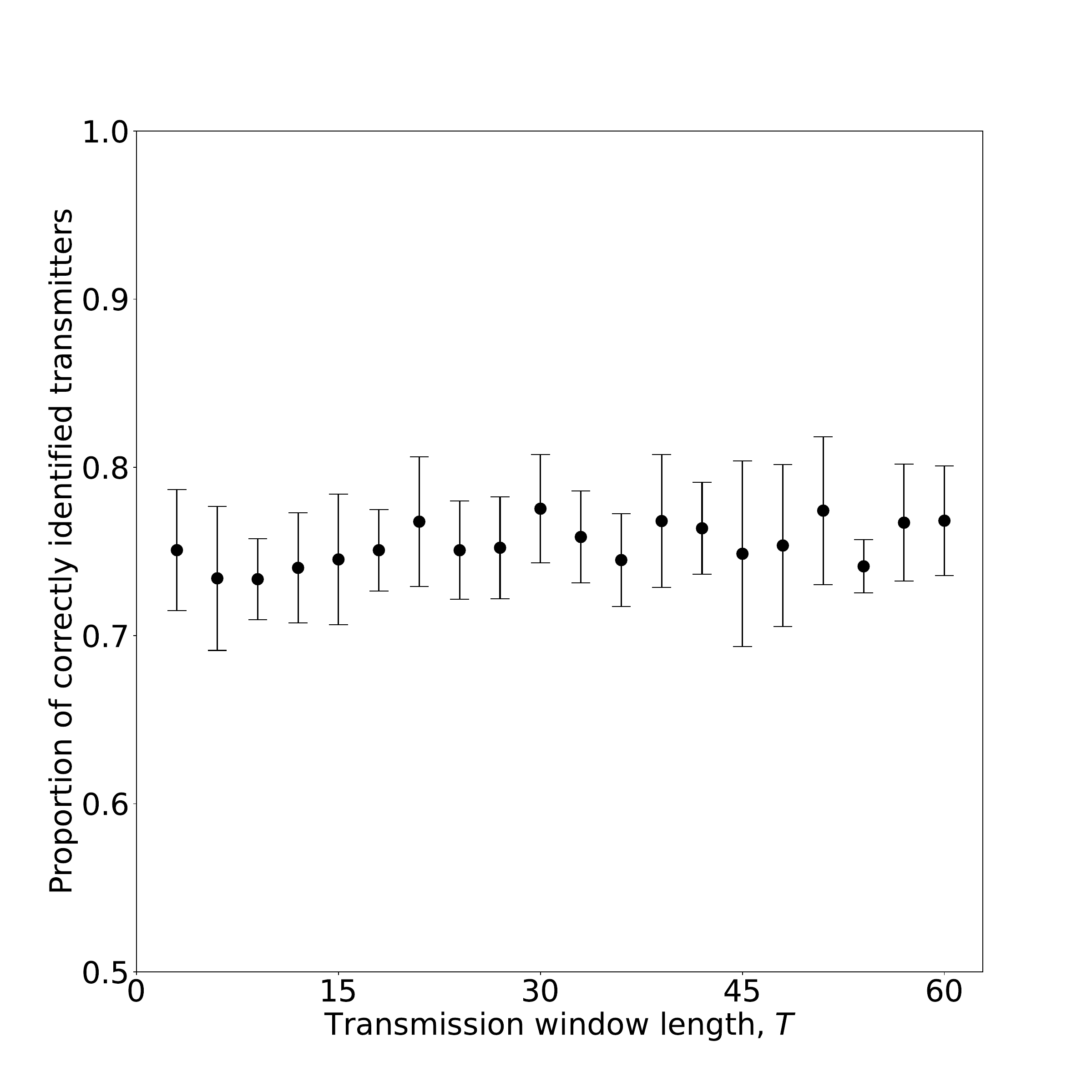}}
    \subfloat[][Correctly identified links]{\includegraphics[width=0.3\textwidth]{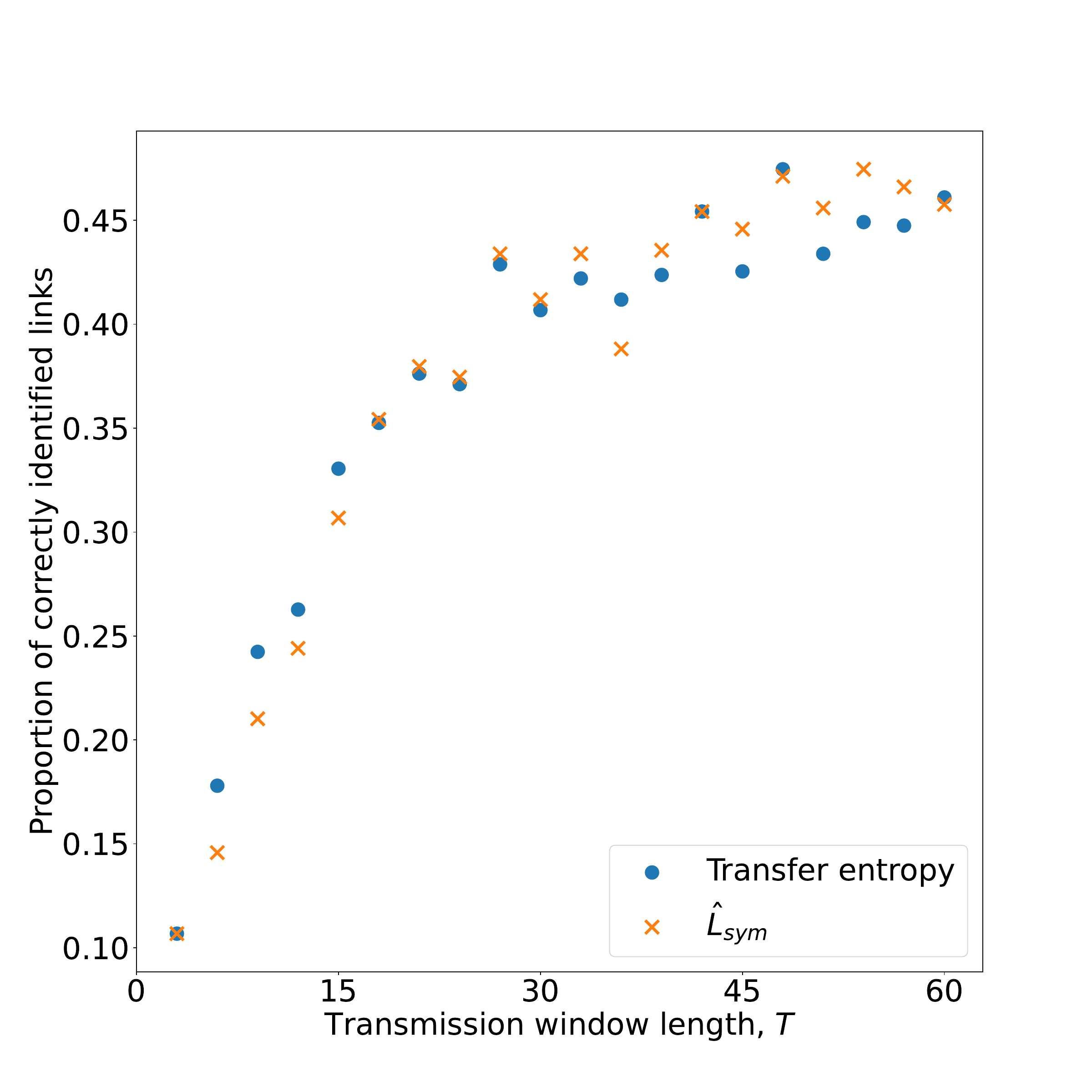}} 
    \subfloat[][Identified vs actual links in the network]{\includegraphics[width=0.3\textwidth]{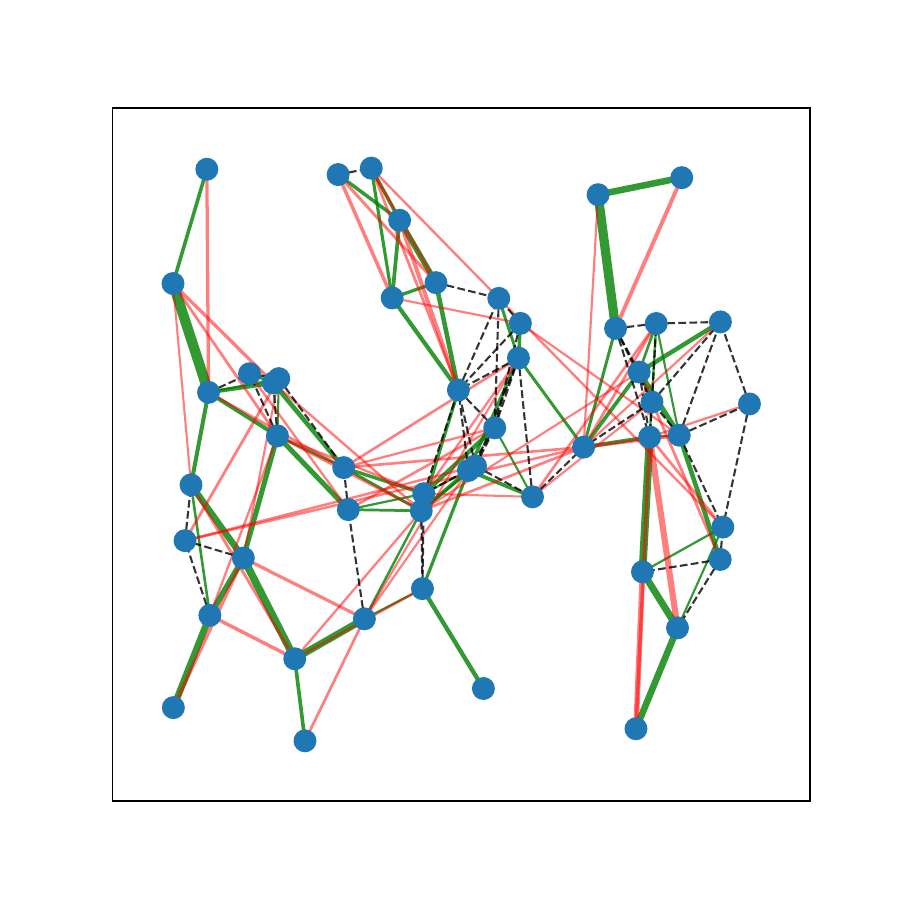}}
 
    \vspace{0.5cm}
    \caption{
    Experimental results for a random network constructed by placing $50$ transmitters and $150$ RF sensors at random in a $340m \times 200m$ rectangular box. 
   Figure (a) displays the adjacency matrix of the network for a random realisation. Figure (b) an estimated matrix $\hat{L}_{\text{sym}}$ for a transmission window length $T=60s$. Figure (c) the corresponding transfer entropy matrix ($T=60s$). 
Figure (d) shows the number of power bursts correctly associated to the corresponding transmitters for a varying transmission window length (we display the average number over 5 random realisations together with a 95\% confidence interval). Figure (e) represents the proportion of correctly identified links by our methodology and transfer entropy, for a varying length of the transmission window (averaged over 5 random realisations). Figure (f) displays (for a specific random realisation with window length $T=60s$) the correctly identified links by $\hat{L}_{\text{sym}}$ in green, and false positives in red; their thickness is proportional to the corresponding entry of $\hat{L}_{\text{sym}}$. Dashed black lines correspond to links in the network that do not appear among the top entries of $\hat{L}_{\text{sym}}$.
    }
    \label{fig:bigbox}
\end{figure*}

\section{CONCLUSION AND FUTURE WORK}
\label{sec:conclusion}
In this work, we have introduced a novel technique for estimating the
topological structure of wireless networks. We demonstrated that our
approach provably recovers information about the network topology under a
simplified Markovian model of network dynamics. Additionally, our
experimental results showed that the proposed method is competitive
with transfer entropy in realistic network simulations, consistently
outperforming it under conditions of network congestion.

We believe that, in addition to being competitive with existing
methods, our approach also provides valuable insights into the
behaviour of traditional techniques. Nevertheless, further research is
needed to enhance both the performance and the theoretical
understanding of our method.
Below, we outline several potential
directions for future work.

First, our method currently requires discretisation in the time domain
and does not account for variable network delays. To address these
limitations, one possible direction would be to consider
continuous-time Markov chains. In this framework, a random clock is
assigned to each node, and a particle moves from the current node once
the clock expires. Since the delays between consecutive packet
transmissions would vary, any estimation algorithm would need to
account for these fluctuating delays.

Another potential generalisation of our model is to employ branching
processes instead of random walks on graphs. In this scenario,
particles could appear, move, die, or duplicate according to specific
probability distributions. This approach could help address the
limitations of our current model, which does not adequately account
for varying traffic levels within the network.

Finally, measures such as transfer entropy and Granger causality are
often used as features in more sophisticated estimation
procedures. For instance, Testi and Giorgetti~\cite{TG20} utilise
these measures to train a neural network model. It would be
interesting to explore the integration of our techniques into machine
learning frameworks for network inference.

\section{PROOFS}
\label{sec:proof}
In this section we provide a proof of our main theoretical result, \cref{thm:main}. Informally, we will show that, for $T$ large enough, $M(u,v) \approx P(u,v) - (k-1)\pi(v)$ and $N(v) \approx T \cdot k \cdot \pi(v)$. Therefore, by computing $M$ and $N$, we are able to approximate $P$. We start by showing that $N(v)$ is approximately a multiple of the stationary distribution at $v$. For technical reason, we also need a bound on the following quantities ($u,v \in V$):
\begin{align}
&\widetilde{N}(u,v) = \sum_{t=1}^{T} \sum_{i\ne j}   \mathbbm{1}\{X_{t}^{(i)} = u, X_{t}^{(j)} = v \}, \label{eq:tildeN} \\
&Q(u,v) = \sum_{t=1}^{T} \sum_{i,i'} \mathbbm{1}\{X_{t}^{(i)} = u, X_{t}^{(i')} = v\}  \label{eq:Qdef}.
\end{align}

We will use the following recent matrix Bernstein inequality for Markov chains~\cite{matrixMarkovBernstein}.

\begin{thm}
\label{thm:matrixMarkovBernstein}
Let $\{X_i\}_{i=1}^\infty$ be a Markov chain with finite state space $V$, stationary distribution $\pi$, and positive spectral gap $\lambda > 0$. Let $f \colon V \to \mathbb{C}^{d \times d}$ such that: (1) $f(u)$ is Hermitian and there exists $A > 0$ such that  $\|f(u)\| \le A$ for all $u \in V$; (2) $\Ex{v \sim \pi}{f(v)} = 0$; (3) $\|\Ex{v \sim \pi} {f(v)^2}\| \le \sigma^2$. Then, for any $\epsilon \in (0,1)$ and $T \in \mathbb{N}$, it holds that
\[
\Pro{\left\| \sum_{t=1}^T f(X_t)\right\| > \epsilon} \le 2 d^2 \exp\left( - \frac{\epsilon^2 \lambda /72 }{T \sigma^2 + A \epsilon} \right).
\]
\end{thm}

We use \cref{thm:matrixMarkovBernstein} to bound the probability $N$, $\tilde{N}$, and $Q$ deviate from their expectation.

\begin{lem}
\label{lem:piapprox}
Let $\epsilon \in (0,1/2), \delta \in (0,1)$. Let $T > C k^2 (\pi^\star/\pi_\star)  \ln(kn/\delta)/(\epsilon^2 \lambda(P) \pi_\star)$ for some large enough universal constant $C > 0$. Let $D_N$ be the $n \times n$ diagonal matrix such that $D_N(u,u) = N(u)$. Then, with probability at least $1-\delta$, for all $v \in V$, the following holds:
\begin{align*}
&\left|N(v) - kT\pi(v)\right| \le \epsilon kT\pi(v), \\
&\left\|\widetilde{N} -  T k(k-1) \pi \pi^T\right\|  \le \epsilon  T k(k-1) \pi_\star \\
&\left\|Q - T k(k-1) \pi \pi^T - D_N\right\| \le T k^2 \pi^\star.
\end{align*}
\end{lem}
\begin{proof}
We first fix $v \in V$ and compute the expectation of $\E{N(v)}$:
\begin{align*}
\E{N(v)} &= \sum_{i=1}^k \sum_{t=1}^{T} \Pro{X^{(j)}_t = v} = kT\pi(v),
\end{align*}
where the last equality holds because we assume the chains start from stationarity.

We now have to bound how far $N(v)$ is from its expectation. We do this by applying the scalar version of Theorem \ref{thm:matrixMarkovBernstein}. Let $f \colon V \to \mathbb{R}$ be defined as $f(u) = \mathbbm{1}\{u = v\}(1- \pi(v))$. Clearly,  $|f(u)| \le 1$ for all $u \in V$, $\Ex{u \sim \pi}{f(u)} = 0$ and $\Ex{v \sim \pi}{f(v)^2} \le \pi(v)$. Notice that $N(v) - \E{N(v)} = \sum_{i=1}^k \sum_{t=1}^T f(X^{(i)}_t)$. Let us fix $i \in \{1,\dots,k\}$. Then, by Theorem \ref{thm:matrixMarkovBernstein} and the assumption on $T$, we have that with probability at most $\delta/(k n)$, $\sum_{t=1}^T f(X^{(i)}_t) \le \epsilon T \pi(v)$. The first inequality in the statement of the lemma follows by taking a union bound over all $v \in V$ and $i \in \{1,\dots,k\}$.


For the second inequality, we first compute $\E{\widetilde{N}(u,v)}$:
\begin{align*}
\E{\widetilde{N}(u,v)} &= \sum_{t=1}^{T} \sum_{i\ne j} \E{  \mathbbm{1}\{X_{t}^{(i)} = u, X_{t}^{(j)} = v \}} \\
&= k(k-1) T \pi(u) \pi(v).
\end{align*}

To control the deviation of $\widetilde{N}(u,v)$ from its expectation, fix $i \ne j$ and consider the product chain $(X^{(i)},X^{(j)})$. It has state space $V \times V$ and stationary distribution $\pi((u,v)) = \pi(u) \pi(v)$. Moreover, its spectral gap is equal to the spectral gap $\lambda(P)$ of the single chain $X^{(i)}$. Consider $g_{ij} \colon V \times V \to \mathbb{R}^{n \times n}$ defined as follows.
\[
g_{ij}((u,v))(x,y) = \begin{cases}
1 - \pi(u)\pi(v) &\text{ if } \{u,v\} = \{x,y\} \\
- \pi(x)\pi(y) &\text{ otherwise.}
\end{cases}
\]
Notice that $\widetilde{N} - \E{\widetilde{N}} = \sum_{i\ne j} \sum_{t=1}^{T} g_{ij}(X_{t}^{(i)}, X_{t}^{(j)})$. Clearly, $\|g_{ij}((u,v))\| \le 2$ for all $(u,v) \in V \times V$. Furthermore, simple calculations show that $\|\Ex{(u,v) \sim \pi \times \pi} {g_{ij}((u,v))^2} \le 2 \pi^{\star}\|$. Therefore, with probability at least $\delta/k^2$, Theorem \ref{thm:matrixMarkovBernstein} implies that $\left\|\sum_{t=1}^{T} g_{ij}(X_{t}^{(i)}, X_{t}^{(j)})\right\| \le Tk\pi_{\star}$. A union bound over all $1 \le i \ne j \le k$ yields the second statement.

Finally, the third inequality follows simply from the first two by noticing that $Q(u,v) = \widetilde{N}(u,v) + D_N(u,v)$.
\end{proof}

We now prove concentration for $M$. 
We apply similar techniques as in Kontorovich and Wolfer~\cite{Wolf19}.  We will need the following concentration inequality for matrix martingales, which is due to Tropp~\cite{Tropp}. 

\begin{thm}
\label{thm:matrixfree}
Consider a random process $\{Y_k \colon k=0,1,2,\dots\}$ such that each $Y_k$ is a $d_1 \times d_2$ matrix. Assume this process satisfies the martingale property:
\[
\Ex{k-1}{Y_k} = Y_{k-1} \quad \text{ for all } k=1,2,\dots,
\]
where with $\Ex{k-1}{Y_k}$ we denote the conditional expectation of $Y_k$ w.r.t. ``the past''.

Let $X_k = Y_k - Y_{k-1}$ for $k=1,2,\dots$, which implies $\{X_k\}_k$ is a martingale difference sequence for $\{Y_k\}_k$, i.e.,
\[
\Ex{k-1}{X_k} = 0 \quad \text{ for all } k=1,2,\dots.
\]

Assume the operator norm of $X_k$ is bounded by $R$ almost surely:
\[
\|X_k\| \le R \quad \text{a.s. for all } k=1,2,\dots
\]
Define $W^{(1)}_k = \sum_{j=1}^k \Ex{j-1}{X_j^*X_j }$ and $W^{(2)}_k = \sum_{j=1}^k \Ex{j-1}{X_jX_j^* }$  for $k=1,2,\dots$. Further assume their operator norm is bounded by $\sigma^2$:
\[
\max\left(\|W^{(1)}_k\|,\|W^{(2)}_k\|\right) \le \sigma^2 \quad \text{a.s. for all } k=1,2,\dots
\]

Then, for all $t > 0$, it holds that
\[
\Pro{\exists k \ge 0 \colon \|Y_k\| \ge t} \le (d_1+d_2) \exp\left( - \frac{t^2/2}{\sigma^2 + Rt/3}\right).
\]
\end{thm}

We define the following random variables. For any $u,v \in V$ and $t=1,2,\dots,T$,
\begin{align}
&Y_t(u,v) = \sum_{i=1}^k \mathbbm{1}\{X_{t-1}^{(i)} = u\}\left( \mathbbm{1}\{X_{t}^{(i)} = v\} - P(u,v)\right) \nonumber \\
		& +  \sum_{1\le i \ne j\le k}  \mathbbm{1}\{X_{t-1}^{(i)} = u\} \mathbbm{1}\{X_{t}^{(j)} = v\} \nonumber \\
		&- \sum_{1\le i \ne j\le k}  \mathbbm{1}\{X_{t-1}^{(i)} = u\}\sum_{w \in V}  \mathbbm{1}\{X_{t-1}^{(j)} = w\} P(w,v). \label{eq:YT}
\end{align}
Notice that 
\[
\Ex{t-1}{ \mathbbm{1}\{X_{t-1}^{(i)} = u\} \mathbbm{1}\{X_{t}^{(i)} = v\}} = \mathbbm{1}\{X_{t-1}^{(i)} = u\} P(u,v)
\]
while, for $i \ne j$,
\begin{align*}
&\Ex{t-1}{ \mathbbm{1}\{X_{t-1}^{(i)} = u\} \mathbbm{1}\{X_{t}^{(j)} = v\} } \\
&\quad =  \mathbbm{1}\{X_{t-1}^{(i)} = u\} \sum_{w \in V}  \mathbbm{1}\{X_{t-1}^{(j)} = w\} P(w,v).
\end{align*}
Therefore,
\[
\Ex{t-1}{Y_t(u,v)} = 0,
\]
which implies that $Y_t$ is a martingale difference sequence.

Let $\widetilde{N}(u,v)=\sum_{t=1}^{T} \sum_{i\ne j}   \mathbbm{1}\{X_{t}^{(i)} = u, X_{t}^{(j)} = v \}$. Then, for any $u,v \in V$,
\begin{align*}
\sum_{t=1}^T Y_t(u,v) &= \left|\{ i,j \in [k], t \in [T] \colon X^{(i)}_{t-1} = u, X^{(j)}_t = v \}\right| \\
	& - \left|\{ i \in [k], 0 \le t < T \colon X^{(i)}_{t} = u\}\right| P(u,v) \\
	& - \sum_{w\in V} \sum_{t=1}^{T} \sum_{i\ne j}   \mathbbm{1}\{X_{t}^{(i)} = u, X_{t}^{(j)} = w \} P(w,v)  \\
	= &N(u,v) - N(u)P(u,v) - \sum_{w \in V} \widetilde{N}(u,w) P(w,v).
\end{align*}

Notice the previous equation can be rewritten in matrix form as
\[
\sum_{t=1}^T Y_t = N - (D_N + \widetilde{N}) P.
\]

Furthermore,
\begin{align*}
\E{\widetilde{N}P(u,v)} &= \sum_{w \in V} \widetilde{N}(u,w) P(w,v)\\
 &= T k(k-1) \pi(u) \sum_{w \in V}\pi(w)P(w,v) \\ 
 &= T k(k-1) \pi(u)\pi(v),
\end{align*}
where the last equality follows because $\pi$ is the stationary distribution of $P$.

This combined with the fact that $\E{N(u)}=Tk\pi(u)$, suggests $\widetilde{M}P(u,v) \approx (k-1)\pi(v)$, where $\widetilde{M} =D_N^{-1} \widetilde{N}$. Indeed, our aim is to prove that $M \approx P + (k-1)\Pi$ where $\Pi(u,v) = \pi(v)$.

By Lemma \ref{lem:piapprox}, with high probability, $\|\widetilde{N} - \E{\widetilde{N}}\| < \epsilon/10$. 
Therefore, 
\begin{align*}
\|\widetilde{N}P - k(k-1) T \pi\pi^T\| &< \|\E{\widetilde{N}}P - k(k-1) T \pi\pi^T\|  \\
&\quad + \|(\widetilde{N} - \E{\widetilde{N}})P\| \\
&\le \|\widetilde{N} - \E{\widetilde{N}}\| \|P\| \\
&\le \frac{Tk\pi_{\star} \epsilon}{10}.
\end{align*}
Hence,
\begin{equation}
\label{eq:NP}
\|D_N(\widetilde{M}P - (k-1) \Pi)\| <  \frac{Tk\pi_{\star} \epsilon}{8}.
\end{equation}

We will apply Theorem \ref{thm:matrixfree} to show that $\|N - (D_N + \widetilde{N}) P\|$ is small. This, together with Equation (\ref{eq:NP}) will imply the theorem.

We start by obtaining a bound on $\|Y_t\|$.

\begin{lem}
\label{lem:YTnorm}
Let $Y_t$ ($t=1,2,\dots$) be defined as in (\ref{eq:YT}). Then, for any $t$, it holds that
\[
\|Y_t\|\le \sqrt{2}k^2.
\]
\end{lem}
\begin{proof}
We bound $\|Y_t\|$ by applying Holder's inequality for matrix norms: $\|Y_t\| \le \sqrt{\|Y_t\|_1\|Y_t\|_{\infty}}$, where $\|Y_t\|_1 = \max_{v \in V} \sum_{u \in V} |Y_t(u,v)|$ and $\|Y_t\|_\infty = \max_{u \in V} \sum_{v \in V} |Y_t(u,v)|$.

We start with a bound on the $\ell_1$ norm.
\begin{align*}
&\|Y_t\|_1 = \max_{v \in V} \sum_{u \in V} |Y_t(u,v)| \\
&\le  \max_{v \in V} \left\{\left|\sum_{i \colon X^{(i)}_{t} = v} (1 - P(X^{(i)}_{t-1},v))\right| \right. \\
&\quad + k \cdot \sum_{i=1}^k \left. \left|\sum_{j \colon X^{(j)}_{t} = v} (1 - P(X^{(j)}_{t-1},v))\right| \right\} \\
& \le k + (k-1)k = k^2.
\end{align*}

For the $\ell_\infty$ norm, first notice that
\begin{align*}
&\sum_{v \in V} \left|\left( \mathbbm{1}\{X_{t}^{(i)} = v\} - P(u,v)\right)\right| \\
&\quad= 1 - P(u,X_{t}^{(i)}) + \sum_{w \ne X_{t}^{(i)}} P(u,w) \\
&\quad= 2(1-P(u,X_{t}^{(i)})).
\end{align*}
Therefore,
\begin{align*}
\|Y_t\|_\infty &= \max_{u \in V} \sum_{v \in V} |Y_t(u,v)| \\
&\le  \max_{u \in V} \left\{2\sum_{i=1}^k (1-P(u,X_{t}^{(i)})) \right. \\
&\quad + \left. \sum_{i=1}^k  \left|\sum_{j\ne i} 2(1 - P(X^{(j)}_{t-1},v))\right| \right\} \\
& \le 2k + 2(k-1)k = 2k^2.
\end{align*}

Hence, 
\begin{equation*}
\|Y_t\| \le \sqrt{\|Y_t\|_1\cdot \|Y_t\|_{\infty}} \le \sqrt{2}k^2.
\end{equation*}
\end{proof}

The following lemma provides a bound on $\sigma^2$ as required by Theorem \ref{thm:matrixfree}. The proof of the lemma is not particularly enlightening but it is relatively involved, since it requires keeping track of different terms arising from the multiplication of $Y_t$ with its transpose.

\begin{lem}
\label{lem:sigmasquare}
Let $Y_t$ ($t=1,2,\dots$) be defined as in (\ref{eq:YT}). Let $W^{(1)}_t = \sum_{j=1}^t \Ex{j-1}{Y_j^TY_j }$ and $W^{(2)}_k = \sum_{j=1}^t \Ex{j-1}{Y_jY_j^T }$. Then, with high probability, it holds that
\begin{equation*}
\|W^{(1)}_T\| \le 8 k^3  T \pi^{\star}.
\end{equation*}
and
\begin{equation*}
\|W^{(2)}_T\| \le 2 k^3  T \pi^{\star}.
\end{equation*}
\end{lem}
\begin{proof}
We first obtain a bound on the operator norm of $W^{(1)}_t$ Let $u,v \in V$. Then,
\begin{align*}
&Y_j^TY_j(u,v) = \sum_{w \in V} Y_j(w,u) Y_j(w,v)\\
&= \sum_{w \in V}    \sum_{i=1}^k \mathbbm{1}\{X_{j-1}^{(i)} = w\}\left( \mathbbm{1}\{X_{j}^{(i)} = u\} - P(w,u)\right) \\
& +	\sum_{w \in V} \sum_{1\le i \ne \ell \le k}  \mathbbm{1}\{X_{j-1}^{(i)} = w\}\left( \mathbbm{1}\{X_{j}^{(\ell)} = u\} - P(X_{j-1}^{(\ell)},u)\right)  \\
&\quad  \cdot \left(\sum_{i'=1}^k \mathbbm{1}\{X_{j-1}^{(i')} = w\}\left( \mathbbm{1}\{X_{j}^{(i')} = v\} - P(w,v)\right)  + \right. \\
&\quad  \left.	 \sum_{1\le i' \ne \ell'\le k}  \mathbbm{1}\{X_{j-1}^{(i')} = w\}\left( \mathbbm{1}\{X_{j}^{(\ell')} = v\} - P(X_{j-1}^{(\ell')},v) \right) \right) \\
&= \sum_{w \in V} \sum_{i, i'}  \mathbbm{1}\{X_{j-1}^{(i)} = X_{j-1}^{(i')} = w\} \left( \mathbbm{1}\{X_{j}^{(i)} = u\} - P(w,u)\right) \\
&\qquad \qquad \cdot \left( \mathbbm{1}\{X_{j}^{(i')} = v\} - P(w,v)\right) \\
&+  \sum_{w \in V} \sum_{i=1}^k \sum_{i'\ne \ell'} \mathbbm{1}\{X_{j-1}^{(i)} = X_{j-1}^{(i')} = w\} \cdot\\
&	\left( \mathbbm{1}\{X_{j}^{(i)} = u\} - P(w,u)\right) \left( \mathbbm{1}\{X_{j}^{(\ell')} = v\} - P(X_{j-1}^{(\ell')},v) \right) \\
&+ \sum_{w \in V} \sum_{i \ne \ell} \sum_{i'}  \mathbbm{1}\{X_{j-1}^{(i)} = X_{j-1}^{(i')} = w\}  \cdot \\
&\left( \mathbbm{1}\{X_{j}^{(\ell)} = u\} - P(X_{j-1}^{(\ell)},u)\right)\left( \mathbbm{1}\{X_{j}^{(i')} = v\} - P(w,v)\right) \\
&+  \sum_{w \in V} \sum_{i\ne \ell} \sum_{i'\ne \ell'} \mathbbm{1}\{X_{j-1}^{(i)} = X_{j-1}^{(i')} = w\}   \\
&\quad \cdot\left( \mathbbm{1}\{X_{j}^{(\ell)} = u\} - P(X_{j-1}^{(\ell)},u)\right) \\
&\quad \cdot \left( \mathbbm{1}\{X_{j}^{(\ell')} = v\} - P(X_{j-1}^{(\ell')},v) \right). 
\end{align*}

We now study its expectation $ \Ex{j-1}{Y_j^TY_j }$. In particular, we split it in four summations as above. We start with the first. Recall $X_{j-1}^{(i)}$ and $X_{j-1}^{(i')}$ are independent for $i \ne i'$. Therefore,
\begin{align*}
&A_j(u,v) \coloneqq \mathbb{E}_{j-1} \Bigl[\sum_{w \in V} \sum_{i, i'}  \mathbbm{1}\{X_{j-1}^{(i)} = X_{j-1}^{(i')} = w\}  \\
& \left( \mathbbm{1}\{X_{j}^{(i)} = u\} - P(w,u)\right) \left( \mathbbm{1}\{X_{j}^{(i')} = v\} - P(w,v)\right)\Bigr] \\
&=\sum_{w \in V} \sum_{i=1}^k  \mathbb{E}_{j-1} \Bigl[\mathbbm{1}\{X_{j-1}^{(i)}  = w\} \left( \mathbbm{1}\{X_{j}^{(i)} = u\} - P(w,u)\right)    \\
&\qquad \cdot \left( \mathbbm{1}\{X_{j}^{(i)} = v\} - P(w,v)\right) \Bigr] \\
&= \sum_{i=1}^k  \sum_{w \in V} \mathbbm{1}\{X_{j-1}^{(i)}  = w\}  \Bigl( \mathbbm{1}\{u=v\}  P(w,u)(1-P(w,u)) \\
&\qquad- \mathbbm{1}\{u \ne v\} P(w,u)P(w,v) \Bigr).
\end{align*}

Analogously, we have that
\begin{align*}
&B_j(u,v) \coloneqq\mathbb{E}_{j-1} \Bigl[\sum_{w \in V} \sum_{i=1}^k \sum_{i'\ne \ell'} \mathbbm{1}\{X_{j-1}^{(i)} = X_{j-1}^{(i')} = w\}  \\
&	\left( \mathbbm{1}\{X_{j}^{(i)} = u\} - P(w,u)\right) \left( \mathbbm{1}\{X_{j}^{(\ell')} = v\} - P(X_{j-1}^{(\ell')},v) \right)\Bigr] \\
& = \sum_{w \in V} \sum_{i \ne i'}\mathbb{E}_{j-1} \Bigl[\mathbbm{1}\{X_{j-1}^{(i)}  = X_{j-1}^{(i')} = w\} \\
&\quad \cdot \left( \mathbbm{1}\{X_{j}^{(i)} = u\} - P(w,u)\right) \\
&\qquad \cdot \left( \mathbbm{1}\{X_{j}^{(i)} = v\} - P(w,v) \right)\Bigr] \\
& = \sum_{w \in V} \sum_{i \ne i'} \mathbbm{1}\{X_{j-1}^{(i)}  = X_{j-1}^{(i')} = w\} \\
&\qquad \cdot \Bigl( \mathbbm{1}\{u=v\}  P(w,u)(1-P(w,u)) \\
&\qquad - \mathbbm{1}\{u \ne v\} P(w,u)P(w,v) \Bigr).
\end{align*}

For the third sum,
\begin{align*}
&C_j(u,v) \coloneqq\mathbb{E}_{j-1} \Bigl[\sum_{w \in V} \sum_{i \ne \ell} \sum_{i'}  \mathbbm{1}\{X_{j-1}^{(i)} = X_{j-1}^{(i')} = w\} \\
&\qquad \cdot \left( \mathbbm{1}\{X_{j}^{(\ell)} = u\} - P(X_{j-1}^{(\ell)},u)\right) \\
&\qquad  \cdot \left( \mathbbm{1}\{X_{j}^{(i')} = v\} - P(w,v)\right) \Bigr] \\
& = \mathbb{E}_{j-1} \Bigl[\sum_{w \in V} \sum_{i \ne \ell}  \mathbbm{1}\{X_{j-1}^{(i)} = X_{j-1}^{(\ell)} = w\} \\
& \qquad \cdot \left( \mathbbm{1}\{X_{j}^{(\ell)} = u\} - P(w,u)\right) \\
&\qquad \cdot \left( \mathbbm{1}\{X_{j}^{(\ell)} = v\} - P(w,v)\right) \Bigr] \\
& =\sum_{w \in V} \sum_{i \ne \ell} \mathbbm{1}\{X_{j-1}^{(i)}  = X_{j-1}^{(\ell)} = w\} \\
&\qquad \cdot \Bigl( \mathbbm{1}\{u=v\}  P(w,u)(1-P(w,u)) \\
&\qquad \quad  - \mathbbm{1}\{u \ne v\}  P(w,u)P(w,v) \Bigr).
\end{align*}

Finally, for the fourth sum,
\begin{align*}
&D_j(u,v) \coloneqq \mathbb{E}_{j-1} \Bigl[\sum_{w \in V} \sum_{i\ne \ell} \sum_{i'\ne \ell'} \mathbbm{1}\{X_{j-1}^{(i)} = X_{j-1}^{(i')} = w\} \\
&\qquad \cdot	\left( \mathbbm{1}\{X_{j}^{(\ell)} = u\} - P(X_{j-1}^{(\ell)},u)\right) \\
&\qquad \cdot \left( \mathbbm{1}\{X_{j}^{(\ell')} = v\} - P(X_{j-1}^{(\ell')},v) \right)\Bigr] \\
& = \mathbb{E}_{j-1} \Bigl[\sum_{w \in V} \sum_{i,i'\ne \ell} \mathbbm{1}\{X_{j-1}^{(i)} = X_{j-1}^{(i')} = w\} \\
&\qquad \cdot	\left( \mathbbm{1}\{X_{j}^{(\ell)} = u\} - P(X_{j-1}^{(\ell)},u)\right) \\
&\qquad \cdot \left( \mathbbm{1}\{X_{j}^{(\ell)} = v\} - P(X_{j-1}^{(\ell)},v)\right)\Bigr] \\
& = \sum_{w,w' \in V} \sum_{i, i' \ne \ell} \mathbbm{1}\{X_{j-1}^{(i)}  = X_{j-1}^{(i')} = w, X_{j-1}^{(\ell)} = w'\} \\
&\qquad \cdot	\Bigl(\mathbbm{1}\{u=v\} P(w',u)(1-P(w',u))  \\
&\qquad \quad   - \mathbbm{1}\{u \ne v\}P(w',u)P(w',v)\Bigr)\\
& = \sum_{w \in V} \sum_{i, i' \ne \ell} \mathbbm{1}\{X_{j-1}^{(i)}  = X_{j-1}^{(i')}, X_{j-1}^{(\ell)} = w\} \\
&\qquad \cdot	 \Bigl( \mathbbm{1}\{u=v\}  P(w,u)(1-P(w,u)) \\
& \qquad \quad   - \mathbbm{1}\{u \ne v\}  P(w,u)P(w,v)\Bigr).
\end{align*}

Notice that we have $W^{(1)}_T = \sum_{t=1}^T (A_t + B_t + C_t + D_t)$. Moreover,
\begin{align*}
\sum_{t=1}^T A_t(u,v) &=  \mathbbm{1}\{u=v\}  \sum_{w \in V} N(w) P(w,u)(1-P(w,u)) \\
&\quad - \mathbbm{1}\{u \ne v\} \sum_{w \in V}N(w)  P(w,u)P(w,v).
\end{align*}
By Lemma \ref{lem:piapprox}, with high probability, we have that $N(u) \le 2 k T \pi(u)$ for any $u \in V$. Therefore,
\begin{align*}
\left\|\sum_{j=1}^T A_t\right\|_1 &\le 2 \max_{u \in V}  \sum_{w \in V} N(w) P(w,u) \\
&\le 4 k T  \max_{u \in V} \sum_{w \in V} \pi(w) P(w,u) \\
&= 4 k T \pi^\star,
\end{align*}
where the last equality follows from the fact that $\pi$ is the stationary measure. Essentially the same calculations show that
\[
\left\|\sum_{j=1}^T A_t\right\|_\infty \le 4 k T \pi^\star,
\]
which implies that
\[
\left\|\sum_{j=1}^T A_t\right\| \le \sqrt{\left\|\sum_{j=1}^T A_t\right\|_1 \left\|\sum_{j=1}^T A_t\right\|_\infty}  \le 4 k T \pi^\star
\]

To bound the contribution to the operator norm given by $B_t$, observe that by Lemma \ref{lem:piapprox}, w.h.p. and for any $u \in V$, $\widetilde{N}(u,u) \le 2 k(k-1) T \pi(u)$.
Therefore, similarly as above, we have that
\begin{align*}
\left\|\sum_{t=1}^T B_t\right\|_1 &\le 2 \max_{u \in V}  \sum_{w \in V} \widetilde{N}(w,w) P(w,u) \\
&\le 4 k(k-1) T  \max_{u \in V} \sum_{w \in V} \pi(w) P(w,u) \\
&\le 4 k(k-1) T \pi^\star 
\end{align*}
and
\[
\left\|\sum_{t=1}^T B_t\right\|_\infty \le 4 k(k-1) T \pi^\star,
\]
which implies that
\[
\left\|\sum_{t=1}^T B_t\right\| \le 4 k(k-1) T \pi^\star.
\]

Analogously, we can show that
\[
\left\|\sum_{t=1}^T C_t\right\| \le 4 k(k-1) T \pi^\star.
\]

In a similar way we bound the contribution given by $D_t$. With high probability, by Lemma \ref{lem:piapprox} it holds that,
\begin{align*}
\left\|\sum_{t=1}^T D_t\right\| &\le  2 k^2 \max_{u \in V}  \sum_{w \in V} N(w) P(w,u)\\
&\le 4 k^3 T \pi^\star.
\end{align*}

By summing the contribution of $A_t, B_t, C_t, D_t$ we prove the first part of the lemma.

We can obtain an analogous bound for $\|W^{(2)}_t\|$. We start by looking at each individual entry $Y_jY_j^T(u,v)$

\begin{align*}
&Y_jY_j^T(u,v) = \sum_{w \in V} Y_j(u,w) Y_j(v,w)\\
&= \sum_{w \in V}  \Bigl(  \sum_{i=1}^k \mathbbm{1}\{X_{j-1}^{(i)} = u\}\left( \mathbbm{1}\{X_{j}^{(i)} = w\} - P(u,w)\right) \\
 &+	 \sum_{1\le i \ne \ell \le k}  \mathbbm{1}\{X_{j-1}^{(i)} = u\}\left( \mathbbm{1}\{X_{j}^{(\ell)} = w\} - P(X_{j-1}^{(\ell)},w)\right) \Bigr) \\
& \cdot \Bigl(\sum_{i'=1}^k \mathbbm{1}\{X_{j-1}^{(i')} = v\}\left( \mathbbm{1}\{X_{j}^{(i')} = w\} - P(v,w)\right)  \\
& +	 \sum_{1\le i' \ne \ell'\le k}  \mathbbm{1}\{X_{j-1}^{(i')} = v\}\left( \mathbbm{1}\{X_{j}^{(\ell')} = w\} - P(X_{j-1}^{(\ell')},w) \right) \Bigr) \\
&= \sum_{w \in V} \sum_{i, i'}  \mathbbm{1}\{X_{j-1}^{(i)} = u, X_{j-1}^{(i')} = v\} \\
&\qquad \cdot \left( \mathbbm{1}\{X_{j}^{(i)} = w\} - P(u,w)\right) \\
&\qquad \cdot \left( \mathbbm{1}\{X_{j}^{(i')} = w\} - P(v,w)\right) \\
&\quad +  \sum_{w \in V} \sum_{i=1}^k \sum_{i'\ne \ell'} \mathbbm{1}\{X_{j-1}^{(i)}  = u, X_{j-1}^{(i')} = v\}  \\
&\qquad	\cdot \left( \mathbbm{1}\{X_{j}^{(i)} = w\} - P(u,w)\right) \\
&\qquad \cdot \left( \mathbbm{1}\{X_{j}^{(\ell')} = w\} - P(X_{j-1}^{(\ell')},w) \right) \\
&\quad+ \sum_{w \in V} \sum_{i \ne \ell} \sum_{i'}  \mathbbm{1}\{X_{j-1}^{(i)} = u, X_{j-1}^{(i')} = v\} \\
&\qquad \cdot \left( \mathbbm{1}\{X_{j}^{(\ell)} = w\} - P(X_{j-1}^{(\ell)},w)\right) \\
&\qquad \cdot \left( \mathbbm{1}\{X_{j}^{(i')} = w\} - P(v,w)\right) \\
&\quad+  \sum_{w \in V} \sum_{i\ne \ell} \sum_{i'\ne \ell'} \mathbbm{1}\{X_{j-1}^{(i)} = u, X_{j-1}^{(i')} = v\} \\
&\qquad	\left( \mathbbm{1}\{X_{j}^{(\ell)} = w\} - P(X_{j-1}^{(\ell)},w)\right) \\
&\qquad \left( \mathbbm{1}\{X_{j}^{(\ell')} = w\} - P(X_{j-1}^{(\ell')},w) \right). 
\end{align*}

We now look at its expectation $ \Ex{j-1}{Y_j^TY_j }$. As done previously, we split it in four summations. By using the independence of $X_{j}^{(i)}$ from $X_{j}^{(i')}$ for $i \ne i'$, we have that
\begin{align*}
A_j'(u,v) &\coloneqq \mathbb{E}_{j-1} \Bigl[ \sum_{w \in V} \sum_{i, i'}  \mathbbm{1}\{X_{j-1}^{(i)} = u, X_{j-1}^{(i')} = v\} \\&\quad \left( \mathbbm{1}\{X_{j}^{(i)} = w\} - P(u,w)\right)\left( \mathbbm{1}\{X_{j}^{(i')} = w\} - P(v,w)\right)\Bigr] \\
&=  \mathbbm{1}\{u=v, X_{j-1}^{(i)} = u\}  \sum_{w \in V}\Bigl(P(u,w)(1-P(u,w))^2  \\
& \qquad+ P(u,w)^2(1-P(u,w))  \Bigr) \\
&= \mathbbm{1}\{u=v, X_{j-1}^{(i)} = u\}  \sum_{w \in V}P(u,w)(1-P(u,w)).
 \end{align*}
 Therefore,
 \begin{align*}
 \sum_{t=1}^T A_t'(u,v) &= \sum_j \mathbbm{1}\{u=v, X_{j-1}^{(i)} = u\} \\
 &\qquad \cdot  \sum_{w \in V} \mathbbm{1}\{X_{j-1}^{(i)} = u\}P(u,w)(1-P(u,w))\\
  &\le \mathbbm{1}\{u=v\} N(u)  \sum_{w \in V} P(u,w) \\
  &= \mathbbm{1}\{u=v\} N(u).
  \end{align*}
Hence,
\[
\left\| \sum_t A_t' \right\| \le \max_u N(u) \le 2 k T \pi^\star.
\]  
  
 Similarly,
 \begin{align*}
&B_j'(u,v) \coloneqq\mathbb{E}_{j-1} \Bigl[  \sum_{w \in V} \sum_{i=1}^k 
\sum_{i'\ne \ell'} \mathbbm{1}\{X_{j-1}^{(i)}  = u, X_{j-1}^{(i')} = v\} \\
&	\left( \mathbbm{1}\{X_{j}^{(i)} = w\} - P(u,w)\right) \left( \mathbbm{1}\{X_{j}^{(\ell')} = w\} - P(X_{j-1}^{(\ell')},w) \right)\Bigr] \\
&= \sum_{w \in V} \sum_{i \ne i'} \mathbbm{1}\{X_{j-1}^{(i)}  = u, X_{j-1}^{(i')} = v\} \\
&\qquad	\left( \mathbbm{1}\{X_{j}^{(i)} = w\} - P(u,w)\right)^2 .
 \end{align*}
  Therefore,
\[
\left\| \sum_t B_t' \right\| \le \left\| \widetilde{N}\right\| \le 2 k(k-1) T \pi^\star.
\]

 Similarly,
 \begin{align*}
&C_j'(u,v) \coloneqq\mathbb{E}_{j-1} \Bigl[  \sum_{w \in V} \sum_{i \ne \ell} \sum_{i'}  \mathbbm{1}\{X_{j-1}^{(i)} = u, X_{j-1}^{(i')} = v\} \\
&\left( \mathbbm{1}\{X_{j}^{(\ell)} = w\} - P(X_{j-1}^{(\ell)},w)\right)\left( \mathbbm{1}\{X_{j}^{(i')} = w\} - P(v,w)\right)\Bigr]  \\
&= \sum_{w \in V} \sum_{i \ne i'} \mathbbm{1}\{X_{j-1}^{(i)} = u, X_{j-1}^{(i')} = v\} \\
&\qquad  \left( \mathbbm{1}\{X_{j}^{(i')} = w\} - P(v,w)\right)^2,
\end{align*}
and 
\[
\left\| \sum_t C_t' \right\| \le 2 k(k-1) T \pi^\star .
\]  

Finally,
 \begin{align*}
&D_j'(u,v) \coloneqq \sum_{w \in V}  \sum_{i,i'\ne \ell} \mathbbm{1}\{X_{j-1}^{(i)} = u, X_{j-1}^{(i')} = v\} \cdot \\
&\qquad		\mathbb{E}_{j-1} \Bigl[ \left( \mathbbm{1}\{X_{j}^{(\ell)} = w\} - P(X_{j-1}^{(\ell)},w)\right)^2	\Bigr] \\
&=  \sum_{w \in V}  \sum_{i,i'\ne \ell} \mathbbm{1}\{X_{j-1}^{(i)} = u, X_{j-1}^{(i')} = v\} \\
&\qquad	\Bigl(P(X_{j-1}^{(\ell)},w)(1-P(X_{j-1}^{(\ell)},w))^2 \\
&\qquad \quad+ (1-P(X_{j-1}^{(\ell)},w))P(X_{j-1}^{(\ell)},w)^2 \Bigr) \\
&=  \sum_{i,i'\ne \ell} \mathbbm{1}\{X_{j-1}^{(i)} = u, X_{j-1}^{(i')} = v\} \\
&\qquad \cdot	\sum_{w \in V}  P(X_{j-1}^{(\ell)},w)(1-P(X_{j-1}^{(\ell)},w)) \\
&\le  \sum_{i,i'\ne \ell} \mathbbm{1}\{X_{j-1}^{(i)} = u, X_{j-1}^{(i')} = v\} 
 \end{align*}
 
  Therefore, by Lemma \ref{lem:piapprox},
  \begin{align*}
 \left\|\sum_{t=1}^T D_t' \right\| \le \left\|Q \right\| \le 3T k^2 \pi^\star.
  \end{align*}

By summing together these four contributions, we prove the second part of the lemma.
\end{proof}

We can now apply Theorem \ref{thm:matrixfree} with $R= \sqrt{2} k^2$ and $\sigma^2 = 8k^3T\pi^\star$. We obtain that
\[
\Pro{\exists k \ge 0 \colon \|Y_k\| \ge r} \le 2n \exp\left( - \frac{r^2/2}{8k^3T\pi^\star + \sqrt{2} k^2 r/3}\right).
\]

In particular, let $r = (1/\sqrt{2}) \cdot Tk\pi_{\star} \epsilon $ and choose $T \ge 5 k^3 \epsilon^{-2}\pi_\star^{-1} \log(n) (4 + \|\pi\|_2^2 \cdot \frac{\pi^\star}{\pi_\star})$, this failure probability will be $O(n^{-5})$.

We have proved that $\|N - (D_N + \widetilde{N}) P\| = \|D_N(M -P -\widetilde{M}P))\| \le (1/\sqrt{2}) \cdot Tk\pi_{\star} \epsilon$ with high probability. By  (\ref{eq:NP}), we have that $\|D_N(M -P - (k-1)\Pi)\| \le (5/6) \cdot Tk\pi_{\star} \epsilon$ with high probability. Since $\min_u D_N(u,u) \ge (9/10) \cdot Tk\pi_{\star}$ (again, with high probability), we have that $\|M - (P + (k-1)\Pi)\| \le \epsilon$. This ends the proof of \cref{thm:main}. 

%


\bibliographystyle{IEEEtran} 
\bibliography{ref} 

\begin{thebibliography}{10}
\providecommand{\url}[1]{#1}
\csname url@samestyle\endcsname
\providecommand{\newblock}{\relax}
\providecommand{\bibinfo}[2]{#2}
\providecommand{\BIBentrySTDinterwordspacing}{\spaceskip=0pt\relax}
\providecommand{\BIBentryALTinterwordstretchfactor}{4}
\providecommand{\BIBentryALTinterwordspacing}{\spaceskip=\fontdimen2\font plus
\BIBentryALTinterwordstretchfactor\fontdimen3\font minus
  \fontdimen4\font\relax}
\providecommand{\BIBforeignlanguage}[2]{{%
\expandafter\ifx\csname l@#1\endcsname\relax
\typeout{** WARNING: IEEEtran.bst: No hyphenation pattern has been}%
\typeout{** loaded for the language `#1'. Using the pattern for}%
\typeout{** the default language instead.}%
\else
\language=\csname l@#1\endcsname
\fi
#2}}
\providecommand{\BIBdecl}{\relax}
\BIBdecl

\bibitem{net_inf_survey}
I.~Brugere, B.~Gallagher, and T.~Y. Berger-Wolf, ``Network structure inference,
  a survey: Motivations, methods, and applications,'' \emph{ACM Computing
  Surveys (CSUR)}, vol.~51, no.~2, pp. 1--39, 2018.

\bibitem{Antonacci20}
\BIBentryALTinterwordspacing
Y.~Antonacci, L.~Astolfi, G.~Nollo, and L.~Faes, ``Information transfer in
  linear multivariate processes assessed through penalized regression
  techniques: Validation and application to physiological networks,''
  \emph{Entropy}, vol.~22, no.~7, 2020. [Online]. Available:
  \url{https://www.mdpi.com/1099-4300/22/7/732}
\BIBentrySTDinterwordspacing

\bibitem{Antonacci24}
Y.~Antonacci, J.~Toppi, A.~Pietrabissa, A.~Anzolin, and L.~Astolfi, ``Measuring
  connectivity in linear multivariate processes with penalized regression
  techniques,'' \emph{IEEE Access}, vol.~12, pp. 30\,638--30\,652, 2024.

\bibitem{TR13}
P.~Tilghman and D.~Rosenbluth, ``Inferring wireless communications links and
  network topology from externals using granger causality,'' in \emph{MILCOM
  2013-2013 IEEE Military Communications Conference}.\hskip 1em plus 0.5em
  minus 0.4em\relax IEEE, 2013, pp. 1284--1289.

\bibitem{LC17}
M.~Laghate and D.~Cabric, ``Learning wireless networks topologies using
  asymmetric granger causality,'' \emph{IEEE Journal of Selected Topics in
  Signal Processing}, vol.~12, no.~1, pp. 233--247, 2017.

\bibitem{TG20}
E.~Testi and A.~Giorgetti, ``Blind wireless network topology inference,''
  \emph{IEEE Transactions on Communications}, vol.~69, no.~2, pp. 1109--1120,
  2020.

\bibitem{liu2022topology}
Z.~Liu, W.~Wang, G.~Ding, Q.~Wu, and X.~Wang, ``Topology sensing of
  non-collaborative wireless networks with conditional granger causality,''
  \emph{IEEE Transactions on Network Science and Engineering}, vol.~9, no.~3,
  pp. 1501--1515, 2022.

\bibitem{granger69}
C.~W. Granger, ``Investigating causal relations by econometric models and
  cross-spectral methods,'' \emph{Econometrica: journal of the Econometric
  Society}, pp. 424--438, 1969.

\bibitem{schreiber00}
T.~Schreiber, ``Measuring information transfer,'' \emph{Physical Review
  Letters}, vol.~85, no.~2, p. 461, 2000.

\bibitem{STsparse}
D.~A. Spielman and S.-H. Teng, ``Spectral sparsification of graphs,''
  \emph{SIAM Journal on Computing}, vol.~40, no.~4, pp. 981--1025, 2011.

\bibitem{ns3}
``ns-3 {N}etwork {S}imulator,'' \url{https://www.nsnam.org}.

\bibitem{hsu19}
D.~Hsu, A.~Kontorovich, D.~A. Levin, Y.~Peres, C.~Szepesv{\'a}ri, and
  G.~Wolfer, ``Mixing time estimation in reversible markov chains from a single
  sample path,'' \emph{Annals of Applied Probability}, vol.~29, no.~4, pp.
  2439--2480, 2019.

\bibitem{Wolf19}
\BIBentryALTinterwordspacing
G.~Wolfer and A.~Kontorovich, ``Estimating the mixing time of ergodic markov
  chains,'' in \emph{Proceedings of the Thirty-Second Conference on Learning
  Theory}, ser. Proceedings of Machine Learning Research, A.~Beygelzimer and
  D.~Hsu, Eds., vol.~99.\hskip 1em plus 0.5em minus 0.4em\relax PMLR, 25--28
  Jun 2019, pp. 3120--3159. [Online]. Available:
  \url{https://proceedings.mlr.press/v99/wolfer19a.html}
\BIBentrySTDinterwordspacing

\bibitem{WolfKont24}
\BIBentryALTinterwordspacing
------, ``{Improved estimation of relaxation time in nonreversible Markov
  chains},'' \emph{The Annals of Applied Probability}, vol.~34, no.~1A, pp. 249
  -- 276, 2024. [Online]. Available: \url{https://doi.org/10.1214/23-AAP1963}
\BIBentrySTDinterwordspacing

\bibitem{salmond19}
D.~Salmond, ``Blind estimation of wireless network topology and throughput,''
  in \emph{2019 53rd Annual Conference on Information Sciences and Systems
  (CISS)}.\hskip 1em plus 0.5em minus 0.4em\relax IEEE, 2019, pp. 1--6.

\bibitem{mehrotra2021minimax}
N.~Mehrotra, E.~Graves, A.~Swami, and A.~Sabharwal, ``Minimax bounds for blind
  network inference,'' in \emph{2021 IEEE International Symposium on
  Information Theory (ISIT)}.\hskip 1em plus 0.5em minus 0.4em\relax IEEE,
  2021, pp. 1823--1828.

\bibitem{du2023network}
W.~Du, T.~Tan, H.~Zhang, X.~Cao, G.~Yan, and O.~Simeone, ``Network topology
  inference based on timing meta-data,'' \emph{IEEE Transactions on
  Communications}, 2023.

\bibitem{ourcode}
\BIBentryALTinterwordspacing
J.~Martin, T.~Pryer, and L.~Zanetti, ``{W}ireless {N}etwork {T}opology
  {I}nference: {A} {M}arkov {C}hains {A}pproach (code and datasets),'' 2025.
  [Online]. Available: \url{http://dx.doi.org/10.5281/zenodo.15592950}
\BIBentrySTDinterwordspacing

\bibitem{markovmixing}
D.~A. Levin and Y.~Peres, \emph{Markov chains and mixing times}.\hskip 1em plus
  0.5em minus 0.4em\relax American Mathematical Soc., 2017, vol. 107.

\bibitem{MonteTetali}
R.~Montenegro, P.~Tetali \emph{et~al.}, ``Mathematical aspects of mixing times
  in markov chains,'' \emph{Foundations and Trends in Theoretical Computer
  Science}, vol.~1, no.~3, pp. 237--354, 2006.

\bibitem{PSZ}
R.~Peng, H.~Sun, and L.~Zanetti, ``Partitioning well-clustered graphs: spectral
  clustering works!'' \emph{SIAM J. Comput.}, vol.~46, no.~2, pp. 710--743,
  2017.

\bibitem{matrixMarkovBernstein}
\BIBentryALTinterwordspacing
J.~Neeman, B.~Shi, and R.~Ward, ``Concentration inequalities for sums of
  markov-dependent random matrices,'' \emph{Information and Inference: A
  Journal of the IMA}, vol.~13, no.~4, p. iaae032, 12 2024. [Online].
  Available: \url{https://doi.org/10.1093/imaiai/iaae032}
\BIBentrySTDinterwordspacing

\bibitem{Tropp}
\BIBentryALTinterwordspacing
J.~Tropp, ``{Freedman's inequality for matrix martingales},'' \emph{Electronic
  Communications in Probability}, vol.~16, no. none, pp. 262 -- 270, 2011.
  [Online]. Available: \url{https://doi.org/10.1214/ECP.v16-1624}
\BIBentrySTDinterwordspacing

\bibitem{PyInform}
\BIBentryALTinterwordspacing
{Douglas G. Moore}, ``Pyinform.'' [Online]. Available:
  \url{https://elife-asu.github.io/PyInform}
\BIBentrySTDinterwordspacing

\bibitem{pinto2024}
H.~Pinto, I.~Lazic, Y.~Antonacci, R.~Pernice, D.~Gu, C.~Bar{\`a}, L.~Faes, and
  A.~P. Rocha, ``Testing dynamic correlations and nonlinearity in bivariate
  time series through information measures and surrogate data analysis,''
  \emph{Frontiers in network physiology}, vol.~4, p. 1385421, 2024.

\end{thebibliography}

\end{document}